\documentclass{lmcs} 
\pdfoutput=1

\usepackage{lastpage}
\lmcsdoi{16}{3}{17}
\lmcsheading{}{\pageref{LastPage}}{}{}%
{Sep.~05,~2019}{Sep.~09,~2020}{}

\usepackage[utf8]{inputenc}

\keywords{functional languages, logical relations, denotational semantics, bar recursion}

\usepackage{amsfonts}
\usepackage{amsmath}
\usepackage{amssymb}
\usepackage{stmaryrd}

\newcommand{\datatype}{\mathtt{datatype}}
\newcommand{\constructor}{\mathtt{cons}}
\newcommand{\func}{\mathtt{func}}
\newcommand{\ar}{\mathrm{ar}}

\newcommand{\csum}[3]{\{{#1},{#2},{#3}\}}
\newcommand{\cplus}[1]{{#1}_+}

\newcommand{\nat}{\mathtt{Nat}}

\newcommand{\zero}{\mathtt{0}}
\newcommand{\suc}{\mathtt{s}}
\newcommand{\rec}{\mathtt{rec}}

\newcommand{\fold}{\mathtt{fold}}
\newcommand{\spec}{\mathtt{spec}}
\newcommand{\sbar}{\mathtt{bar}}

\newcommand{\emp}{\varepsilon}
\newcommand{\cons}{\; \mathtt{::} \;}
\newcommand{\len}{\mathtt{len}}
\newcommand{\ext}{\mathtt{ext}}
\newcommand{\num}[1]{\mathbf{#1}}

\newcommand{\sm}[1]{\bt{#1}}
\newcommand{\dss}[1]{\overline{#1}}
\newcommand{\ds}[1]{\llbracket {#1} \rrbracket}
\newcommand{\bt}[1]{\langle {#1} \rangle}
\newcommand{\mt}[1]{|{#1}|}

\newcommand{\initSeg}[2]{[{#1}]({#2})}

\newcommand{\btl}{\blacktriangleleft}
\newcommand{\wtl}{\lhd}
\newcommand{\bnd}{\btl}

\def\NN{\mathbb{N}}

\newcommand{\modset}{\mathcal{S}^\omega}
\newcommand{\modcont}{\mathcal{C}^\omega}

\newcommand{\inc}{\mathrm{inc}}
\newcommand{\comp}{\mathrm{com}}
\newcommand{\cl}[1]{\mathtt{Cl}_{#1}}
\newcommand{\val}[1]{\mathtt{Val}_{#1}}

\newcommand{\red}{R}
\newcommand{\vred}{R^{\scriptsize{\mathrm{val}}}}

\newcommand{\bP}{\mathcal{P}}

\newcommand{\mlambda}{%
  \mathop{%
    \rlap{$\lambda$}%
    \mkern2mu
    \raisebox{.275ex}{$\lambda$}%
  }%
}

\newcommand{\dplus}{%
  \mathbin{{+}\mspace{-8mu}{+}}%
}

\newcommand{\maj}[1]{ \; \mathrm{maj}_{#1} \; }

\theoremstyle{plain} 

\begin{document}

\title{A unifying framework for continuity and complexity in higher types}

\author[T.~Powell]{Thomas Powell}	
\address{Department of Computer Science, University of Bath, Bath, BA2 7AY, United Kingdom}	
\email{trjp20@bath.ac.uk}  







\begin{abstract}
We set up a parametrised monadic translation for a class of call-by-value functional languages, and prove a corresponding soundness theorem. We then present a series of concrete instantiations of our translation, demonstrating that a number of fundamental notions concerning higher-order computation, including termination, continuity and complexity, can all be subsumed into our framework. Our main goal is to provide a unifying scheme which brings together several concepts which are often treated separately in the literature. However, as a by-product, we also obtain (i) a method for extracting moduli of continuity for closed functionals of type $(\nat\to\nat)\to\nat$ in (extensions of) System T, and (ii) a characterisation of the time complexity of bar recursion. 
\end{abstract}

\maketitle

\section{Introduction}\label{sec-intro}

Monads are a fundamental tool for analysing functional programs \cite{Moggi(1991.0),Wadler(1992.0)}. They allow us to capture information about a program's execution, such as its \emph{computational complexity}, and similarly enable us to reason about intensional aspects of higher-order functionals, such as \emph{continuity} properties enjoyed by terms of System T.

This paper comprises a general study of monads and their application to higher-order functional languages, with an emphasis on languages which pertain to proof theory and program extraction. We focus on a simple yet powerful technique frequently encountered in the literature, which can be roughly described as follows:
\begin{enumerate}[(I)]

\item set up a mapping $\mt{\cdot}$ which transforms terms $e$ of type $\rho$ in some pure functional language into terms $\mt{e}$ of type $\mt{\rho}$ in some monadic metalanguage,

\item establish a logical relation $\bnd$ between $\rho$ and $\mt{\rho}$,

\item appeal to induction over the structure of terms to prove that $e\bnd\mt{e}$ for all terms $e$ of type $\rho$.

\end{enumerate}
We explore an abstract translation of this kind, which maps a higher-order call-by-value target language into a metalanguage based on the \emph{writer} monad $C\times -$. Our mapping will actually be the composition of a syntactic monadic translation $\mt{\cdot}$ together with a standard denotational semantics $\ds{\cdot}$, and we define a general logical relation $\bnd$ on $\rho\times\ds{\mt{\rho}}$ which comprises a monadic part $\Downarrow$ together with a semantic bounding relation $\wtl$. 

Our first main result is to establish a general soundness theorem for $\bnd$, which identifies a set of abstract conditions that guarantee that $e\bnd \ds{\mt{e}}$ for all terms of our target language. We then demonstrate that this combination of a monadic and bounding component is rich enough to express a variety of fundamental concepts from the theory of higher-order computation, including \emph{majorizability} (in the sense of Howard \cite{Howard(1973.0)}), \emph{continuity} (on the Baire space) and several notions of \emph{time complexity}.

In the case of continuity, we focus on terms of type level 2 that induce functionals $\NN^\NN\to\NN$. It is well known that functionals of this sort that are definable as closed terms of G\"{o}del's System T are continuous in the sense that their output is determined by a finite part of their input. Moreover, explicit \emph{moduli} of continuity can be defined within System T itself. In recent years there has been a renewed interest in syntactic approaches to continuity which provide explicit moduli of continuity, and several new proofs of the continuity of System T definable functionals have been given, including \cite{CoqJab(2012.0),Escardo(2013.0),Hernest(2007.1)} and most recently \cite{Kawai(2019.0),BicRah(2018.0),Xu(2020.0)}. We provide another proof of this fact via our abstract framework, which applies not just to System T but to arbitrary languages that satisfy the relevant conditions.

We then show that by adjusting the parameters of our translation, we are able to characterise the time complexity, or \emph{cost}, of terms in our programming language. Here, we crucially consider the cost of a higher-order term to be a higher-order object itself. Though the analysis of higher-order complexity via higher-order cost expressions dates back to the 1980s and has been widely researched since (e.g.~\cite{Benzinger(2004.0),Sands(1990.0),Shultis(1985.0),VanStone(2003.0)}), our general bounding relation $\wtl$ allows us to also approach complexity along the lines of Danner et al.~\cite{DanLicRam(2015.0),DanPayRoy(2013.0)}, where datatypes are assigned abstract sizes and the translation seeks to provide upper bounds on the cost of programs.

The main significance of our work lies in our development an abstract semantics for reasoning about intensional properties of higher-order functionals, which encompasses a number of somewhat disparate concepts, ranging from majorizability, which is of fundamental importance to the proof mining program (Section \ref{sec-simple-maj}), to the static cost analysis of functional programs (Section \ref{sec-cost}). This effort towards unification is not just of theoretical interest, but could potentially inform those working on the \emph{formalization} of monadic translations\footnote{For example, in the case of continuity, both \cite{Escardo(2013.0)} and \cite{Xu(2020.0)} have been formalised in Agda, while bounded complexity as set out in \cite{DanPayRoy(2013.0)} comes with a Coq implementation}.

A second novelty is that we focus not only on variants of G\"{o}del's System T, but also more complex forms of recursion such as Spector's \emph{bar recursion} \cite{Spector(1962.0)}. Recursion of this kind is rarely treated in the context of static analysis, but has great importance in proof theory and particularly the area of program extraction, where it is used to give a computational interpretation to the axiom of countable choice. We carry out what is, to the best of our knowledge, the first static cost analysis of bar recursion as functional programs, and in doing so we hope more generally to provide another illustration of how monads enable us to better understand structures from proof theory.

\section{A higher-order functional language}
\label{sec-prog}

We start by outlining our target language, to which our monadic translation will be applied. This will be a standard call-by-value typed functional language, where for now we do not specify what our datatypes or function symbols will be. Rather we take these as parameters, and consider a number of concrete instantiations of the language later on.

Types and terms of the language are outlined in Figure \ref{fig-typeterms}. Types are built from a base set of datatypes $\datatype$, and a single rule which allows for the construction of function types. Terms are then built in the usual way via lambda abstraction and function application from a countable set of variables $x:\rho$ for each type together with a set of constructor terms $\constructor$ and a set of function symbols $\func$. Each constructor term $c$ is assigned some type $\delta_1\to\ldots\to\delta_k\to\delta$, where the $\delta_i$ are datatypes and $k=\ar(c)$ is the arity of $c$. Similarly, each function symbol $f$ is assigned a type $\rho_1\to\ldots\to\rho_k\to\rho$, where now the $\rho_i,\rho$ are arbitrary and $k=\ar(f)$ is a specified arity of $f$. Formal typing rules $\Gamma\vdash t:\rho$ for terms are included, though we often just write $t:\rho$ or $t^\rho$ where this is unambiguous and the context is not necessary. A closed term $e$ is a term without any free variables, or alternatively one typeable as $\emptyset\vdash e:\rho$.

\begin{figure}[t]
\begin{gather*}\mbox{Types:} \ \ \ \rho,\tau::=\delta\; | \; \rho\to\tau \ \ \ \mbox{for} \ \ \ \delta\in\datatype\\
\mbox{Terms:} \ \ \ s,t::=x\; | \; c \; | \; f \; | \; \lambda x.t \; | \; ts \ \ \ \mbox{for} \ \ \ c\in\constructor,f\in\func \\[3mm] \Gamma,x:\rho\vdash x:\rho \ \ \ \ \ \ \frac{\Gamma,x:\rho\vdash t:\tau}{\Gamma\vdash \lambda x.t:\rho\to\tau} \ \ \ \ \ \  \frac{\Gamma\vdash t:\rho\to\tau \ \ \ \ \Gamma\vdash s:\rho}{\Gamma\vdash ts:\tau}\\[2mm]
\Gamma\vdash c:\delta_1\to\ldots\to\delta_m\to \delta \ \ \ \ \ \ \Gamma\vdash f:\rho_1\to\ldots\to\rho_n\to \rho
\end{gather*}
\caption{Types and terms of target language}
\label{fig-typeterms}
\hrulefill
\end{figure}

In order to specify the operational semantics of our language, we introduce standard notions of \emph{patterns} and \emph{values}, which are defined as in Figure \ref{fig-patval}. Note that values are always closed terms - in the third clause this is ensured by the typing restriction that $r$ has no free variables other than $x:\rho$. We also need the notion of a \emph{substitution}, and later that of a \emph{value environment}. For $\Gamma,x:\rho\vdash t:\tau$ and $\Gamma\vdash s:\rho$ we write $\Gamma\vdash t[s/x]$ to denote the term obtained by substituting all free occurrences of $x$ in $t$ by $s$ (with the usual restrictions on free variables in $s$). This is formally definable by induction on the structure of $t$. For a term $\Gamma\vdash t:\rho$ with $\Gamma:\equiv x_1:\rho_1,\ldots,x_n:\rho_n$, a value environment $\sigma$ is a mapping which assigns each variable $x_i:\rho_i$ a value $\sigma(x_i)$ of type $\rho_i$. We denote by $t\sigma:\rho$ the closed term $t[\sigma(x_1)/x_1,\ldots,\sigma(x_k)/x_k]$.

\begin{figure}[t]
\begin{equation*}
\begin{aligned}&\mbox{Patterns:} \ \ \ p,q::=x^\rho\; | \; cp_1\ldots p_k \ \ \ \mbox{for} \ \ \ k=\ar(c)\\
&\mbox{Values:} \ \ \ u,v::= c v_1\ldots v_i \; | \; fv_1\ldots v_j \; | \; \lambda x.r \ \ \ \mbox{for} \ \ \ i\leq \ar(c), \ j< \ar(f), \ x:\rho\vdash r \end{aligned}\end{equation*}
\caption{Patterns and values}
\label{fig-patval}
\hrulefill
\end{figure}

A big step operational semantics for our language is given in Figure \ref{fig-bigstep}.
\begin{figure}[t]
\begin{gather*}
v\downarrow_\rho v\mbox{ \ \ for values $v$} \\[2mm] 
\frac{r[u/x]\downarrow_\tau v}{(\lambda x.r)u\downarrow_\tau v}\mbox{ \ \ for $x:\rho\vdash r$ and $u:\rho$ a value} \\[2mm]
\frac{r\sigma\downarrow_\rho v}{fv_1\ldots v_n\downarrow_\rho v}\mbox{ \ \ for $v_1,\ldots,v_n=p_1\sigma,\ldots,p_n\sigma$ and $fp_1\ldots p_n\leadsto r$}\\[2mm]
 \frac{e\downarrow_{\rho\to\tau} u \ \ \ \ \ e'\downarrow_\rho v \ \ \ \ \ \ uv\downarrow_\tau w}{ee'\downarrow_\tau w}\mbox{ \ \ \ if one of $e$, $e'$ is not a value}
\end{gather*}
\caption{Operational semantics of target language}
\label{fig-bigstep}
\hrulefill
\end{figure}
Note that reductions are of the form $e\downarrow_\rho v$ where $e:\rho$ is a closed term and $v:\rho$ a value of the same type. We often omit the typing on $\downarrow_\rho$ when there is no risk of ambiguity. Reductions follow the usual rules for the call-by-value lambda calculus, together with a set of defining rules for each function symbol of the form
\begin{equation*}
fp_1\ldots p_k\leadsto r
\end{equation*} 
where $p_1,\ldots,p_k$ are patterns, $k=\ar(f)$ and $r$ is a term whose free variables are contained in those of $p_1,\ldots, p_k$. We assume that the rules defining each function symbol are complete and orthogonal, by which we mean that for each $f$ of arity $k$ and any values $v_1,\ldots,v_k$ of the appropriate type, there is exactly one rule $fp_1\ldots p_k\leadsto r$ such that $v_1,\ldots,v_n=p_1\sigma,\ldots,p_k\sigma$ for a suitable environment $\sigma$. Note that for any closed term $e$, at most one rule in Figure \ref{fig-bigstep} is applicable, and thus a simple induction over derivations proves that if $e\downarrow v$ and $e\downarrow v'$ then $v=v'$.

\begin{rem}
\label{rem-alt-target}
A more traditional presentation of our target language would have been to instead take the formation rule
\begin{equation*}
\frac{\Gamma\vdash t_1:\rho_1 \ \ \ \cdots \ \ \ \Gamma\vdash t_n:\rho_n}{\Gamma\vdash ft_1\ldots t_n:\rho} 
\end{equation*}
as primitive, together with an analogous rule for the constructors. Then values would simply be terms of the form $cv_1\ldots v_m$ (for $m=\ar(c)$) or $\lambda x.r$, and the semantics would be altered accordingly. However, we have instead chosen to take the function \emph{symbols} as primitive, as it is slightly more convenient for the purposes of setting up our monadic translation.
\end{rem}

We now give some concrete instantiations of our parametrised language, all of which will play a role later.

\subsection{G\"{o}del's System T (simple variant)}
\label{sec-prog-syst}

A call-by-value variant of System T is obtained in our setting by defining $\datatype:=\{\nat\}$, $\constructor:=\{\zero^\nat,\suc^{\nat\to\nat}\}$ and 
\begin{equation*}
\func:=\{\rec^{\rho\to(\nat\to\rho\to\rho)\to\nat\to\rho}_\rho\; : \; \mbox{$\rho$ a type}\}
\end{equation*}
where the recursor $\rec_\rho$ has defining equations
\begin{equation*}
\rec_\rho\; x\; y\; \zero\leadsto x \ \ \ \ \ \ \rec\; x \; y \; (\suc z)\leadsto y\; z\; (\rec_\rho \; x \; y \; z).
\end{equation*}
We write $\num{n}:=\suc^{(n)}(\zero)$ for the numeral representation of $n\in\NN$. In this language, one can show by induction on the size of $v$ that the only values of type $\nat$ are indeed those of the form $\num{n}$, a fact which we will freely use throughout. The operational semantics of $\rec_\rho$ can be concisely expressed via the following derived rules:
\begin{equation*}
\rec_\rho\; v_1\; v_2\; \zero\downarrow_\rho v_1 \ \ \ \ \ \ \frac{v_2\; \num{n}\downarrow_{\rho\to\rho}u \ \ \ \ \ \ \rec_\rho\; v_1\; v_2\; \num{n}\downarrow_\rho v \ \ \ \ \ \  uv\downarrow_\rho w}{\rec_\rho\; v_1\; v_2\; \suc(\num{n})\downarrow_\rho w}
\end{equation*}

\subsection{G\"{o}del's System T (list based variant)}
\label{sec-prog-lyst}
A slightly richer (but computationally equivalent) version of System T is obtained by defining $\datatype:=\{\nat,\nat^\ast\}$, $\constructor:=\{\zero^\nat,\suc^{\nat\to\nat},\emp^{\nat^\ast},\cons^{\nat^\ast\to\nat\to\nat^\ast}\}$ and
\begin{equation*}
\func:=\{\fold_\rho^{\rho\to (\nat\to\rho\to\rho)\to\nat^\ast\to \rho}\; : \; \mbox{$\rho$ a type}\}\cup\{+,\times,<,\len\}
\end{equation*}
where $+,\times,<$ have type $\nat\to\nat\to\nat$ and $\len$ has type $\nat^\ast\to\nat$. We adopt the usual convention of writing $s::t$ instead of $::\; s\; t$, and similarly for the operator symbols, and in addition we write $[s_1,\ldots,s_{k}]$ for $(\ldots (\varepsilon::s_1)::\dots )::s_{k}$. For $a=(a_1,\ldots,a_{k})\in\NN^\ast$ we write $\num{a}:=[\num{a}_1,\ldots,\num{a}_{k}]$. Again, a simple induction over the size of values establishes that any $v:\nat^\ast$ is of the form $v=\num{a}$ for some $a\in\NN^\ast$. Along with the usual defining equation for the fold function:
\begin{equation*}
\fold_\rho\; x\; y\; \emp\leadsto x \ \ \ \ \ \ \fold_\rho \; x\; y\; (zs\cons z)\leadsto y\; z\; (\fold_\rho\; x\; y\; zs) 
\end{equation*}
which gives rise to the derived rules
\begin{equation*}
\fold_\rho\; v_1\; v_2\; \varepsilon\downarrow_\rho v_1 \ \ \ \ \ \ \frac{v_2\; \num{n}\downarrow_{\rho\to\rho} u \ \ \ \ \ \ \fold_\rho\; v_1\; v_2\; \num{a}\downarrow_\rho v \ \ \ \ \ \ uv\downarrow_\rho w}{\fold_\rho\; v_1\; v_2\; (\num{a}::\num{n})\downarrow_\rho w}
\end{equation*}
we can incorporate basic operators into our language by defining each of them explicitly via a countable set of rules e.g.~$\num{m}+\num{n}\leadsto\num{k}$ for $k=m+n$ and so on, which then satisfy operational rules $\num{m}+\num{n}\downarrow_\nat\num{k}$. We define $<$ and $\len$ in an analogous fashion so that
\begin{equation*}
\begin{aligned}
&\num{m}<\num{n}\downarrow_\nat\zero\mbox{ \ \ if $m<n$, else \ \ }\num{m}<\num{n}\downarrow_\nat \num{1}.\\
&\len\; \num{a}\downarrow_\nat \num{k}\mbox{ \ \ for $k=|a|$}
\end{aligned}
\end{equation*}
where $|a|$ denotes the length of $a\in\NN^\ast$.

\subsection{Spector's bar recursion}
\label{sec-prog-spec}

The final language we consider here is an extension of the list language in Section \ref{sec-prog-lyst} with Spector's bar recursor of lowest type, originally introduced in \cite{Spector(1962.0)} (but see e.g.~\cite{Oliva(2006.2)} for a more modern introduction). This is a form of backward recursion over wellfounded trees, where for this particular variant, wellfoundedness is typically ensured by appealing to some form of continuity on the parameters of the recursion. In a simple equational calculus it would be given by the defining equation
\begin{equation*}
B(\omega,\phi,\psi,a)=_{\NN}\begin{cases}\phi(a) & \mbox{if $\omega(\hat{a})<|a|$}\\ \psi(a,\lambda x\; . \; B(\omega,\phi,\psi,a\ast x)) & \mbox{otherwise}\end{cases}
\end{equation*}
where the output type is $\NN$ as indicated, and the input parameters have type $\omega:\NN^\NN\to\NN$, $\phi:\NN^\ast\to\NN$ and $\psi:\NN^\ast\times(\NN\to\NN)\to\NN$ and $a\in \NN^\ast$ respectively. In addition, we have $a\ast x:=(a_1,\ldots,a_k,x)$ and $\hat{a}:\NN\to\NN$ is defined by $\hat{a}_n:=a_i$ for $i<n$ and otherwise $0$. Bar recursion will primarily play a role in Sections \ref{sec-simple-ds} and \ref{sec-cost-bound}, where we focus on termination and complexity, respectively. For now, we simply give the main definitions we will need later. We first need to add three additional function symbols
\begin{equation*}
\begin{aligned}
\ext &: \nat^\ast\to\nat\to\nat\\
\sbar &: \rho_1\to\rho_2\to\rho_3\to\nat^\ast\to\nat\\
\sbar_1 &: \rho_1\to\rho_2\to\rho_3\to \nat^\ast\to\nat\to\nat
\end{aligned}
\end{equation*}
where
\begin{equation*}
\begin{aligned}
\rho_1&:=(\nat\to\nat)\to\nat\\
\rho_2&:=\nat^\ast\to\nat\\
\rho_3&:=\nat^\ast\to (\nat\to\nat)\to\nat
\end{aligned}
\end{equation*}
which have defining equations
\begin{equation*}
\begin{aligned}
\ext\; [x_1,\ldots,x_k]\; \num{n}&\leadsto \begin{cases}x_n & \mbox{if $n<k$}\\ \zero & \mbox{otherwise}\end{cases} \\
\sbar\; f\; g\; h\; xs&\leadsto \sbar_1\; f\; g\; h\; xs\; (f(\ext\; xs)<\len\; xs) \\
\sbar_1\; f\; g\; h\; xs\; \zero&\leadsto g\; xs\\
\sbar_1\; f\; g\; h\; xs\; \suc z&\leadsto h\; xs\; (\lambda x\; . \; \sbar\; f\; g\; h\; (xs::x))
\end{aligned}
\end{equation*}
It is not difficult to show that the operational semantics in this case give rise to the following derived rules for $\sbar$ (here we eliminate the intermediate steps involving $\sbar_1$)
\begin{gather*}
\frac{v_1(\ext\; \num{a})\downarrow_\nat \num{k} \ \ \ \ \ \ k<|a| \ \ \ \ \ \ v_2\; \num{a}\downarrow_\nat \num{n}}{\sbar\; v_1\; v_2\; v_3\; \num{a}\downarrow \num{n}}\\[2mm]
\frac{v_1(\ext\; \num{a})\downarrow_\nat \num{k} \ \ \ \ \ \ k\geq |a| \ \ \ \ \ \ v_3\; \num{a}\; (\lambda x\; . \; \sbar\; v_1\; v_2\; v_3\; (\num{a}::x))\downarrow_\nat \num{n}}{\sbar\; v_1\; v_2\; v_3\; \num{a}\downarrow_\nat \num{n}}
\end{gather*}
\begin{rem}
\label{rem-bar}
The above formulation of bar recursion as a rewrite system is inspired by Berger \cite{Berger(2006.0)}, which in turn uses a trick due to Vogel \cite{Vogel(1976.0)}. There bar recursion is considered in its general form, where the type of $\num{a}$ and the output can be arbitrary. Though such a generalisation can easily be incorporated here by encoding lists of type $\rho$ as objects of type $\nat\to\rho$, for us bar recursion plays a predominantly illustrative role, and so for simplicity we restrict ourselves to the recursor of base type.
\end{rem}

\section{The main soundness theorem}
\label{sec-sound}

In this section we present our main framework, first setting up our monadic translation and the associated logical relation, before proving that the translation is sound with respect to the relation. This soundness result, given as Theorem \ref{thm-sound}, is just an instance of the General Theorem of Logical Relations for our particular relation. As mentioned right at the beginning, our translation is actually the composition of a syntactic translation into a monadic metalanguage and a standard denotational semantics. We outline each of these in turn.

\subsection{The monadic metalanguage}
\label{sec-sound-meta}

The metalanguage is defined similarly to our target language. However, here we only specify types and terms, which will then be assigned a denotational semantics in the next section. The monadic language is summarised in Figure \ref{fig-ptypeterms}. 
\begin{figure}[t]
\begin{gather*}
\mbox{Types:} \ \ \ \rho,\tau::=\gamma \; | \; \mt{\delta} \; | \; \rho\times \tau \; | \; \rho\to \tau \ \ \ \mbox{for} \ \ \ \delta\in\datatype\\
\mbox{Terms:} \ \ \ r,s,t::=\iota \; | \; \cplus{t} \; | \; \csum{r}{s}{t} \; | \; x\; | \; \bt{c} \; | \; \bt{f} \; | \; \lambda x.t \; | \; ts \; | \; (s,t) \; | \; t_l \; | \; t_r \ \ \mbox{for} \ \ \ c\in\constructor,f\in\func \\[3mm] \Gamma\vdash \iota:\gamma \ \ \ \ \ \ \frac{\Gamma\vdash t:\gamma}{\Gamma\vdash \cplus{t}:\gamma} \ \ \ \ \ \  \frac{\Gamma\vdash r:\gamma \ \ \ \Gamma\vdash s:\gamma \ \ \ \Gamma\vdash t:\gamma}{\Gamma\vdash \csum{r}{s}{t}:\gamma} \\[3mm] \frac{\Gamma\vdash s:\rho \ \ \ \Gamma\vdash t:\tau}{\Gamma\vdash (s,t):\rho\times \tau} \ \ \ \ \ \ \frac{\Gamma\vdash t:\rho\times \tau}{\Gamma\vdash t_l:\rho} \ \ \ \ \ \ \frac{\Gamma\vdash t:\rho\times \tau}{\Gamma\vdash t_r:\tau} \\[3mm]
\Gamma\vdash \bt{c}:\mt{\delta_1}\to\ldots\to\mt{\delta_m}\to\mt{\delta} \ \ \ \ \ \ \Gamma\vdash \bt{f}:\mt{\rho_1}\to\ldots\to\mt{\rho_n}\to \gamma\times\mt{\rho} \end{gather*}
\caption{Types and terms of metalanguage}
\label{fig-ptypeterms}
\hrulefill
\end{figure}
Types are constructed from a special type $\gamma$ together with base types $\mt{\delta}$ for all datatypes of the target language, and now allow both function types and cartesian product types. We extend the mapping $\mt{\cdot}$ on datatypes of the target language to \emph{arbitrary} types by defining
\begin{equation*}
\mt{\rho\to\tau}:=\mt{\rho}\to\gamma\times\mt{\tau}.
\end{equation*} 
Terms now include a symbol $\bt{c}$ resp. $\bt{f}$ for each constructor resp. function symbol in the target language, whose types are indicated in Figure \ref{fig-ptypeterms}. Together with the usual term forming rules of the lambda calculus (now with pairing and projection to account for the product), we have three nonstandard rules for forming terms of type $\gamma$: a constant $\iota$, a unary operation $\cplus{t}$ and a ternary operation $\csum{r}{s}{t}$. The meaning of these will become clearer later on. 

For the first step of our translation, we assign to each $\Gamma\vdash t:\rho$ in the target language a term $\mt{\Gamma}\vdash \mt{t}:\gamma\times\mt{\rho}$ in the metalanguage as indicated in Figure \ref{fig-target-meta}, where for $\Gamma=x_1:\rho_1,\ldots,x_k:\rho_k$, we define $\mt{\Gamma}:=x_1:\mt{\rho_1},\ldots,x_k:\mt{\rho_k}$. Here we also use the following shorthand: for $r:\gamma\times \mt{\tau}$ we define 
\begin{equation*}
\lambda^\ast x^\rho.r:=(\iota,\lambda x^\rho.r):\gamma\times\mt{\rho\to\tau}
\end{equation*}
and similarly for the iterated version $\lambda^\ast x_1^{\rho_1},\ldots,x_k^{\rho_k} . r:\gamma\times\mt{\rho_1\to\ldots\to\rho_k\to\tau}$.
\begin{figure}[t]
\begin{equation*}
\begin{aligned}
\Gamma\vdash t:\rho &\mapsto \mt{\Gamma}\vdash\mt{t}:\gamma\times\mt{\rho}\\[3mm]
\mt{x}&:=(\iota,x)\\
\mt{\lambda x.t}&:=\lambda^\ast x\; . \; (\cplus{\mt{t}_l},\mt{t}_r)\\
\mt{ts}&:=(\csum{\mt{t}_l}{\mt{s}_l}{(\mt{t}_r\mt{s}_r)_l},(\mt{t}_r\mt{s}_r)_r)\\
\mt{c}&:=\lambda^\ast x_1,\ldots,x_m\; . \; (\iota,\bt{c}x_1\ldots x_m)\\
\mt{f}&:=\lambda^\ast x_1,\ldots,x_n\; . \; \bt{f}x_1\ldots x_n
\end{aligned}
\end{equation*}
\caption{Translation from target to metalanguage}
\label{fig-target-meta}
\hrulefill
\end{figure}

\subsection{A denotational semantics for the metalanguage}
\label{sec-sound-ds}

We assign our metalanguage a denotational semantics as specified in Figure \ref{fig-denotational}. We assign sets $C$ and $D_\delta$ to the base types $\gamma$ and $\mt{\delta}$ respectively, for each $\delta\in\datatype$. We leave open the precise meaning of the function space $X\Rightarrow Y$ for now. In what follows we will typically work in \emph{total} models, where function spaces will either be the full set-theoretic function space, or the space of continuous functions from $X$ to $Y$ in the sense of Kleene \cite{Kleene(1959.0)} or Kreisel \cite{Kreisel(1959.0)}

\begin{figure}[t]
\begin{equation*}
\ds{\gamma}:=C \ \ \ \ \ \  \ds{\mt{\delta}}:=D_\delta \ \ \ \ \ \ \ds{\rho\times\tau}:=\ds{\rho}\times\ds{\tau} \ \ \ \ \ \ \ds{\rho\to\tau}:=\ds{\rho}\Rightarrow\ds{\tau}
\end{equation*}
\begin{equation*}
\begin{aligned}
\ds{\iota}\xi&:=\epsilon\\
\ds{\cplus{t}}\xi&:=\inc(\ds{t}\xi)\\
\ds{\csum{r}{s}{t}}\xi&:=\comp(\ds{r}\xi,\ds{s}\xi,\ds{t}\xi)\\
\ds{x}\xi&:=\xi(x)\\
\ds{\sm{c}}\xi&:=\dss{c}\\
\ds{\sm{f}}\xi&:=\dss{f}\\
\ds{\lambda x.t}\xi&:=\mlambda a.\ds{t}\xi\{x\mapsto a\}\\
\ds{ts}\xi&:=\ds{t}\xi(\ds{s}\xi)\\
\ds{(s,t)}\xi&:=(\ds{s}\xi,\ds{t}\xi)\\
\ds{t_l}\xi&:=\pi_0(\ds{t}\xi)\\
\ds{t_r}\xi&:=\pi_1(\ds{t}\xi)
\end{aligned}
\end{equation*}
\caption{Denotational semantics of metalanguage}
\label{fig-denotational}
\hrulefill
\end{figure}

Terms $\Gamma\vdash t:\rho$ of the metalanguage are assigned an interpretation $\ds{t}\xi\in\ds{\rho}$ in the usual way, where $\xi$ is an environment mapping each $x_i:\rho_i\in\Gamma$ to some $\xi(x_i)\in \ds{\rho_i}$. We assume that we have chosen suitable interpretations $\dss{c}$, $\dss{f}$ for each $\bt{c}$ and $\bt{f}$ respectively. Terms of type $\gamma$ are interpreted via a triple $(\epsilon^C,\inc^{C\to C},\comp^{C\times C\times C\to C})$.

\subsection{The logical relation}
\label{sec-sound-log}

Let $\rho$ be a type of our target language. We denote the set of all closed terms of type $\rho$ by $\cl{\rho}$, and the set of all values by $\val{\rho}$. We now suppose that we are given a pair of relations 
\begin{equation*}
\mbox{$\Downarrow_\rho$ on $\cl{\rho}\times C\times \val{\rho}$ and $\wtl_\delta$ on $\val{\delta}\times \ds{\mt{\delta}}$,}
\end{equation*}
where $\rho$ ranges over all types in our target language and $\delta$ over all datatypes. We will generalise $\wtl_\rho$ to arbitrary types via the inductive clause
\begin{equation*}
u\wtl_{\rho\to\tau} f:\Leftrightarrow (\forall v\in\val{\rho},a\in\ds{\mt{\rho}})(v\wtl_\rho a\Rightarrow uv\bnd_\tau f(a))
\end{equation*}
where the relation $\bnd_\rho$ on $\cl{\rho}\times (C\times\ds{\mt{\rho}})$ is defined by
\begin{equation*}
e\bnd_\rho a:\Leftrightarrow (\exists v\in\val{\rho})(e\Downarrow^{\pi_0a}_\rho v\wedge v\wtl_\rho \pi_1a),
\end{equation*}
where in what follows we often abbreviate the inner conjunction as $e\Downarrow^{\pi_0a} v\wtl \pi_1 a$. In the remainder of this section we seek to establish conditions under which the translation $\ds{\mt{\cdot}}$ is sound with respect to $\bnd$, by which we mean the following:
\begin{defi}
\label{def-sound}
Given a value environment $\sigma$ for $\Gamma:=x_1:\rho_1,\ldots,x_k:\rho_k$ in our target language and some denotational environment $\xi$ for $\mt{\Gamma}=x_1:\mt{\rho_1},\ldots,x_k:\mt{\rho_k}$, we write $\sigma\wtl_\Gamma \xi$ if $\sigma(x_i)\wtl_{\rho_i} \xi(x_i)$ for all $i=1,\ldots,k$. We say that $\ds{\mt{\cdot}}$ is sound w.r.t. $\bnd$ if for all $\Gamma\vdash t:\rho$ in our target language, we have
\begin{equation*}
\sigma\wtl_\Gamma\xi \Rightarrow t\sigma\bnd_\rho \ds{\mt{t}}\xi
\end{equation*}
for all $\sigma,\xi$.
\end{defi}
We start by focusing on the monadic part of our relation.
\begin{defi}
\label{def-compmon}
We say that $(\epsilon,\inc,\comp)$ is compatible with $\Downarrow$ if the following rules are satisfied:
\begin{gather*}
v\Downarrow^\epsilon_\rho v\mbox{ \ \ for values $v$} \\[2mm] 
\frac{r[u/x]\Downarrow^c_\tau v}{(\lambda x.r)u\Downarrow^{\inc(c)}_\tau v}\mbox{ \ \ for $x:\rho\vdash r$ and $u:\rho$ a value} \\[2mm]
 \frac{e\Downarrow^{c_0}_{\rho\to\tau} u \ \ \ \ \ e'\Downarrow^{c_1}_\rho v \ \ \ \ \ \ uv\Downarrow^{c_2}_\tau w}{ee'\Downarrow^{\comp(c_0,c_1,c_2)}_\tau w}
\end{gather*}
\end{defi}

\begin{lem}
\label{lem-compapp}
Suppose that $(\epsilon,\inc,\comp)$ is compatible with $\Downarrow$. For $f\in C\times (\ds{\mt{\rho}}\to C \times\ds{\mt{\tau}})$ and $a\in C\times \ds{\mt{\rho}}$, define
\begin{equation*}
f\circ a:=(\comp(\pi_0f,\pi_0a,\pi_0((\pi_1f)(\pi_1a))),\pi_1((\pi_1f)(\pi_1a)))\in C\times\ds{\mt{\tau}}.
\end{equation*}
Then whenever $e\bnd_{\rho\to\tau} f$ and $e'\bnd_\rho a$ then $ee'\bnd_{\tau} f\circ a$.
\end{lem}

\begin{proof}
We have $e\Downarrow^{\pi_0f}u\wtl_{\rho\to\tau}\pi_1 f$ and $e'\Downarrow_\rho^{\pi_0a} v\wtl \pi_1a$ for some $u,v$, and thus $uv\btl_\tau (\pi_1 f)(\pi_1 a)$, or in other words $uv\Downarrow_{\tau}^{\pi_0((\pi_1 f)(\pi_1 a))} w\wtl_\tau \pi_1((\pi_1 f)(\pi_1 a))$ for some $w$. Thus by Definition \ref{def-compmon} we have $ee'\Downarrow^{\pi_0(f\circ a)}_\tau w\wtl_\tau \pi_1(f\circ a)$, which is just $ee'\bnd_\tau f\circ a$.
\end{proof}
\begin{defi}
We say that our interpretations $\dss{c}$ for the constructor terms are compatible with $\wtl$ if for all constructor symbols $c:\delta_1\to\ldots\to\delta_m\to\delta$ we have $cv_1\ldots v_m\wtl_{\delta} \dss{c}a_1\ldots a_m$ whenever $v_i\wtl_{\delta_i}a_i$ for all $i=1,\ldots,m$ and values $v_i$.
\end{defi}
\begin{thm}
\label{thm-sound}
Suppose that $(\epsilon,\inc,\comp)$ is compatible with $\Downarrow$ and the $\dss{c}$ are compatible with $\wtl$. Suppose in addition that for all function symbols $f:\rho_1\to\ldots\to\rho_n\to\rho$ we have $fv_1\ldots v_n\bnd_\rho \dss{f}a_1\ldots a_n$ whenever $v_i\wtl_{\rho_i}a_i$ for all $i=1,\ldots,n$ and values $v_i$. Then $\ds{\mt{\cdot}}$ is sound w.r.t. $\bnd$.
\end{thm}

\begin{proof}
We use induction on the structure of terms $\Gamma\vdash t:\rho$ to prove that $t\sigma\btl_\rho\ds{\mt{t}}\xi$ whenever $\sigma\wtl_\Gamma\xi$. For variables, $x\sigma\bnd_\rho \ds{\mt{x}}\xi=(\epsilon,\xi(x))$ follows from the fact that $\sigma(x)\Downarrow^\epsilon \sigma(x)$ and $\sigma(x)\wtl_\rho \xi(x)$. For constructor symbols, suppose that $c:\delta_1\to\ldots\to\delta_m\to\delta$. Then we have $c\Downarrow^\epsilon_\rho c$, and since $\ds{\mt{c}}=(\epsilon,\ds{\lambda x_1\lambda^\ast x_2,\ldots,x_m.(\iota,\bt{c}x_1\ldots x_m)})$ it remains to show that for $v_1\wtl_{\delta_1} a_1$ we have
\begin{equation*}
cv_1\bnd\ds{\lambda^\ast x_2,\ldots,x_m.(\iota,\bt{c}x_1\ldots x_m)}\{x_1\mapsto a_1\}.
\end{equation*}
Continuing this way for $cv_1v_2,\ldots$ we must ultimately show that
\begin{equation*}
cv_1\ldots v_m\bnd \ds{(\iota,\bt{c}x_1\ldots x_m)}\{x_1,\ldots,x_m\mapsto a_1,\ldots a_m\}=(\epsilon,\bar{c}a_1\ldots a_m)
\end{equation*}
which follows from the fact that $\dss{c}$ is compatible with $\wtl$. This argument is easily adapted to show that $f\bnd\ds{\mt{f}}$, using the main assumption of the theorem. 

For function application, by the induction hypothesis we have $t\sigma\bnd_{\rho\to \tau}\ds{\mt{t}}\xi$ and $s\sigma\bnd_\rho \ds{\mt{s}}\xi$. Thus by Lemma \ref{lem-compapp} we have
\begin{equation*}
(ts)\sigma=(t\sigma)(s\sigma)\bnd_\tau \ds{\mt{t}}\xi\circ\ds{\mt{s}}\xi=\ds{\mt{ts}}\xi
\end{equation*}
where the last equality follows from unwinding definitions. In remains to deal with abstraction. Suppose that $\sigma\wtl\xi$. We need to show that
\begin{equation*}
(\lambda x.t)\sigma=(\lambda x.t\sigma)\bnd_{\rho\to\tau}\ds{\lambda^\ast x\; .\; (\cplus{\mt{t}_l},\mt{t}_r)}\xi=(\epsilon,\mlambda a\; . \; (\inc(\pi_0(\ds{\mt{t}}\xi_a)),\pi_1(\ds{\mt{t}}\xi_a)))
\end{equation*}
for $\xi_a:=\xi\{x\mapsto a\}$. But since $\lambda x.t\sigma\Downarrow^\epsilon_{\rho\to\tau}\lambda x.t\sigma$ this reduces to showing that for any $u\wtl_\rho a$ we have
\begin{equation*}
(\lambda x.t\sigma)u\Downarrow_\tau^{\inc(\pi_0(\ds{\mt{t}}\xi_a))}v\wtl_\tau \pi_1(\ds{\mt{t}}\xi_a)
\end{equation*}
for some $v\in\val{\tau}$. Because $u\wtl_\rho a$ implies $\sigma\{x\mapsto u\}\wtl \xi_a$, by the induction hypothesis we have
\begin{equation*}
t\sigma\{x\mapsto u\}=t\sigma[u/x]\Downarrow_\tau^{\pi_0(\ds{\mt{t}}\xi_a)}v\wtl_\tau \pi_1(\ds{\mt{t}}\xi_a)
\end{equation*}
for some $v$, and thus the result follows from Definition \ref{def-compmon}.\end{proof}
This concludes the first main part of the article, where we introduce our target language and establish a sound monadic translation which acts on it. Up to this point everything has been fairly standard: The first component of our translation into the monadic metalanguage is a simple call-by-value monadic translation using the writer monad. Theorem \ref{thm-sound} is then a confirmation that the logical relation defined on terms of our target language and denotations of our metalanguage acts as it should. We now move onto the applications, where we show that Theorem \ref{thm-sound}, though simple, is surprisingly versatile.

\section{Some simple applications for $C=\{0\}$}
\label{sec-simple}

Before we move on to our main applications of the translation, we demonstrate that in the simple case where $C=\{0\}$ is a terminal object, thus collapsing the monadic part of the translation, our soundness theorem can still be related to a number of key concepts in the literature. In each of the examples that follow, we define
\begin{equation*}
e\Downarrow_\rho^{0} v:\Leftrightarrow e\downarrow v
\end{equation*}
Note that in this case, $(\epsilon,\inc,\comp)$ are uniquely defined as just constant functions, and are compatible with $\Downarrow$ since the conditions of Definition \ref{def-compmon} follow from the definition of $\downarrow$ in Figure \ref{fig-bigstep}.

\subsection{Reducibility predicates}
\label{sec-simple-red}

Let's first consider the case where in addition to $C=\{0\}$ we have $D_\delta=\{0\}$ for all datatypes, and thus $\ds{\rho}\cong\{0\}$ for all types of the metalanguage - where here we implicitly use the isomorphisms $\{0\}\times\{0\}\cong\{0\}$ and $(X\to \{0\})\cong \{0\}$. We define
\begin{equation*}
v\wtl_\delta 0:\Leftrightarrow \mathsf{T}.
\end{equation*}
Now let us write $\red_\rho(e):\Leftrightarrow e\btl_\rho 0$ and $\vred_\rho(v):\Leftrightarrow v\wtl_\rho 0$. Then it is easy to show that
\begin{equation*}
\red_\delta(e)\Leftrightarrow (\exists v)(e\downarrow v) \ \ \ \ \ \ \red_{\rho\to\tau}(e):\Leftrightarrow (\exists u)(e\downarrow u\wedge (\forall v)(\vred_\rho(v)\Rightarrow \red_\tau(uv)))
\end{equation*}
and therefore $\red(e)$ and $\vred(v)$ act as reducibility predicates. Moreover, whatever our constant symbols are, the $\bar{c}\in\{0\}$ are trivially compatible with $\wtl$. Thus the following result follows directly from our abstract soundness theorem.
\begin{thm}
\label{thm-red}
Suppose that all function symbols $f:\rho_1\to\ldots\rho_n\to\rho$ in our target language satisfy
\begin{equation*}
\vred_{\rho_1}(v_1)\wedge\ldots\wedge \vred_{\rho_n}(v_n)\Rightarrow \red_\rho(fv_1\ldots v_n)
\end{equation*}
for all values $v_1,\ldots v_n$ of the appropriate type. Then for any closed term $e:\rho$ there exists some value $v$ such that $e\downarrow v$.
\end{thm}
The above result confirms that termination of our target language as a whole follows from termination of the function symbols. As a simple application, we prove that our call-by-value version of System T (Section \ref{sec-prog-syst}) terminates.
\begin{cor}
System T is terminating.
\end{cor}

\begin{proof}
We need to show that for any $v_1,v_2$ of the appropriate type and numeral $\num{n}$ we have
\begin{equation*}
\vred_\rho(v_1)\wedge \vred_{\nat\to \rho\to\rho}(v_2)\Rightarrow\red_\rho(\rec_\rho \; v_1 \; v_2\; \num{n})
\end{equation*}
(note that $\vred_\nat(\num{n})$ trivially holds). We do this by induction on $n\in\NN$, using the derived rules for the recursor given in Section \ref{sec-prog-syst}. For $n=0$ we have $\rec_\rho\; v_1\; v_2\; \zero\downarrow v_1$ and since $\vred_\rho(v_1)$ holds by assumption it follows that $\red_\rho(\rec_\rho\; v_1\; v_2\; \zero)$. Now for the induction step, we assume that $\red_\rho(\rec_\rho \; v_1\; v_2\; \num{n})$ for some $n\in\NN$, which means that $\rec_\rho \; v_1\; v_2\; \num{n}\downarrow v$ for some $v$ satisfying $\vred_\rho(v)$. Since $\vred_{\nat\to\rho\to\rho}(v_2)$ implies that $v_2\; \num{n}\downarrow u$ for some $u$ with $\vred_{\rho\to\rho}(u)$, we then have $R_\rho(uv)$, which means that $uv\downarrow_\rho w$ for some $w$ with $\vred_\rho(w)$. Finally, by the derived rule for the recursor, we have $\rec_\rho\; v_1\; v_2\; \suc(\num{n})\downarrow w$ and thus $\red_\rho(\rec_\rho\; v_1\; v_2\; \suc(\num{n}))$, and we're done.
\end{proof}

For the usual formulation of System T as an equational calculus, the above result roughly corresponds to the fact that any inner-most, left-most sequence of reductions terminates in some normal form (cf. Troelstra \cite[Theorem 2.2.6]{Troelstra(1973.0)}). It readily generalises to functional languages with more complex datatypes and other forms of wellfounded recursion. However, for languages involving bar recursion, such as that presented in Section \ref{sec-prog-spec}, termination requires us to make use of our bounding component $\wtl$.

\subsection{Termination via denotational semantics}
\label{sec-simple-ds}

We now instantiate $\wtl_\rho$ so that it mimics the logical relation used in Plotkin's famous adequacy proof for PCF \cite{Plotkin(1977.0)}, and demonstrate that this allows us to prove termination of bar recursion. The key here is to interpret our metalanguage in a \emph{continuous} model, and then appeal to properties of the model in order to establish termination. Our approach closely follows Berger \cite{Berger(2005.0),Berger(2006.0)}, who sets up a more general framework in which \emph{strong} normalization of higher-order rewrite systems can be proven using domain theoretic means.

We work with the instance of our target language given in Example \ref{sec-prog-spec}, which contains two datatypes $\nat$ and $\nat^\ast$. We define $C=\{0\}$ as before, but now let $D_\nat:=\NN$ and $D_{\nat^\ast}:=\NN^\ast$. In addition, we interpret the function space $X\Rightarrow Y$ in our denotational model as the space of all \emph{total continuous functionals} from $X$ to $Y$. Thus our metalanguage is interpreted in the standard model $\modcont$ of total continuous functionals over base types $\NN$ and $\NN^\ast$, which has various presentations in the literature: Via limit spaces (Scarpellini \cite{Scarpellini(1971.0)}), encodings or neighbourhoods (Kleene/Kreisel \cite{Kleene(1959.0),Kreisel(1959.0)}) or as the extensional collapse of the partial continuous functionals (Ershov \cite{Ershov(1977.0)}).

We first observe that for any type $\rho$ in our target language, $\ds{\mt{\rho}}$ is isomorphic to the usual (non-monadic) denotational semantics of that type, since 
\begin{equation*}
\ds{\mt{\rho\to\tau}}=\ds{\mt{\rho}}\Rightarrow \{0\}\times \ds{\mt{\tau}}\cong \ds{\mt{\rho}}\Rightarrow \ds{\mt{\tau}}
\end{equation*}
and thus we can ignore the monadic part of our translation. We define our logical relation as in Section \ref{sec-simple-red}, but now with a semantic component as follows:
\begin{equation*}
\num{n}\wtl_\nat m:\Leftrightarrow n=m \ \ \ \ \ \ \num{a}\wtl_{\nat^\ast} b:\Leftrightarrow a=b
,\end{equation*}
and as such, $\btl_\rho$ takes the following simplified form:
\begin{equation*}
\begin{aligned}
e\btl_\rho a&:\Leftrightarrow (\exists v)(e\downarrow v\wedge v\wtl_\rho a) \\
u\wtl_{\rho\to\tau} f&:\Leftrightarrow (\forall v,a)(v\wtl_\rho a\Rightarrow uv\btl_\tau f(a)). 
\end{aligned}
\end{equation*}
Interpreting the constructor symbols in the obvious way i.e.
\begin{equation*}
\dss{0}:=0 \ \ \ \ \ \ \dss{\suc}(n):=n+1 \ \ \ \ \ \ \dss{\varepsilon}:=[] \ \ \ \ \ \ \dss{::}(a,n):=a\ast n
\end{equation*}
it is clear that these interpretations are compatible with $\wtl$. In addition, we interpret the operator function symbols of our target language in the usual way, where in particular i.e.
\begin{equation*}
\dss{\len}(a):=|a| \ \ \ \ \ \ \dss{\ext}(a):=\hat{a}.
\end{equation*}
Finally, for bar recursion we let
\begin{equation*}
\begin{aligned}
\dss{\sbar}(\omega,g,h,a)&=\begin{cases}g(a) & \mbox{if $\omega(\hat{a})<|a|$}\\ h(a,\mlambda x\;. \; \dss{\sbar}(\omega,g,h,a\ast x)) & \mbox{otherwise}\end{cases}
\end{aligned}
\end{equation*}
Note that $\dss{\sbar}$ is just the traditional defining equation of bar recursion, and is known to be an object of $\modcont$ as originally proven by Scarpellini \cite{Scarpellini(1971.0)}. We now prove that the \emph{program} $\sbar$ is terminating.
\begin{lem}
\label{lem-barterm}
Whenever $v_1\wtl_{(\nat\to\nat)\to\nat}\omega$, $v_2\wtl_{\nat^\ast\to\nat}g$, $v_3\wtl_{\nat^\ast\to (\nat\to\nat)\to\nat} h$ and $a\in\NN^\ast$, we have
\begin{equation*}
\begin{aligned}
\sbar\; v_1\; v_2\; v_3\; \num{a}\btl_\nat\dss{\sbar}(\omega,g,h,a).
\end{aligned}
\end{equation*}
\end{lem}

\begin{proof}
Assuming the hypotheses of the lemma throughout, we first claim that for all $a\in\NN^\ast$, whenever
\begin{equation*}
(\ast) \ \ \ \sbar\; v_1\; v_2\; v_3\; \num{a}::\num{n}\btl_\nat\dss{\sbar}(\omega,g,h,a\ast n)
\end{equation*}
for all $n\in\NN$ then $\sbar\; v_1\; v_2\; v_3\; \num{a}\btl_\nat\dss{\sbar}(\omega,g,h,a)$. There are two cases to consider:
\begin{itemize}

\item Case 1: $\omega(\hat{a})<|a|$. Then since $\ext\; \num{a}\wtl \hat{a}$ and $v_1\wtl\omega$ we have $v_1(\ext\; \num{a})\btl_\nat \omega(\hat{a})$ and thus $v_1(\ext\; \num{a})\downarrow \num{k}$ for $k=\omega(\hat{a})<|a|$. Since in addition $v_2\wtl g$ and thus $v_2\; \num{a}\btl_\nat g(a)$ it follows that $v_2\; \num{a}\downarrow \num{n}$ for $n=g(a)$, and thus by the first derived rule given in Section \ref{sec-prog-spec} we have $\sbar\; v_1\; v_2\; v_3\; \num{a}\downarrow \num{n}$ for $n=\dss{\sbar}(\omega,g,h,a)$.

\item Case 2: $\omega(\hat{a})\geq |a|$. By the same reasoning, we have $v_1(\ext\; \num{a})\downarrow\num{k}$ for $k=\omega(\hat{a})\geq |a|$. Now, by our assumption $(\ast)$ we have
\begin{equation*}
(\sbar\; v_1\; v_2\; v_3\; \num{a}::x)[\num{n}/x]\btl_\nat \dss{\sbar}(\omega,g,h,a::x)\{x\mapsto n\}
\end{equation*}
for any $n$, and since $\num{n}\wtl m$ iff $n=m$ it follows that
\begin{equation*}
\lambda x\; . \; \sbar\; v_1\; v_2\; v_3\; \num{a}::x\wtl_{\nat\to\nat} \mlambda x\; . \; \dss{\sbar}(\omega,g,h,a::x).
\end{equation*}
Using in addition that $v_3\wtl h$ it is not hard to show that 
\begin{equation*}v_3\; \num{a}\; (\lambda x\; . \; \sbar\; v_1\; v_2\; v_3\; \num{a}::x)\btl_\nat h(a,\mlambda x \; . \; \dss{\sbar}(\omega,g,h,a::x))
\end{equation*}
and thus $v_3\; \num{a}\; (\lambda x\; . \; \sbar\; v_1\; v_2\; v_3\; \num{a}::x)\downarrow \num{n}$ for $n=\dss{\sbar}(\omega,g,h,a::x))$. Thus by the second derived rule we have $\sbar\; v_1\; v_2\; v_3\; \num{a}\downarrow \num{n}$ for $n=\dss{\sbar}(\omega,g,h,a)$.

\end{itemize}
Combining cases we obtain $\sbar\; v_1\; v_2\; v_3\; \num{a}\btl \dss{\sbar}(\omega,g,h,a)$, which proves the claim. We now come to the crucial part of the proof in which we utilise properties of the model $\modcont$ to establish termination of bar recursive programs. Suppose that for some $a\in\NN^\ast$ it is not the case that $\sbar\; v_1\; v_2\; v_3\; \num{a}\btl \dss{\sbar}(\omega,g,h,a)$. Then using our claim (in its contrapositive form) together with dependent choice on a metalevel (which is permitted in $\modcont$) there is an infinite sequence of numbers $b:=b_0,b_1,b_2,\ldots$ such that
\begin{equation*}
\neg (\sbar \; v_1\; v_2\; v_3\; (\num{a}::\num{b}_0::\cdots::\num{b}_{l-1})\btl \dss{\sbar}(\omega,g,h,a\ast\initSeg{b}{l}))
\end{equation*}
for all $l\in\NN$, where $\initSeg{b}{l}:=(b_0,\ldots,b_{l-1})$ denotes the initial segment of $b$ of length $l$. But it can be shown using a standard continuity argument (cf. \cite{Scarpellini(1971.0)}) that for any $\omega:\NN^\NN\to\NN$ that there exists some $L\in\NN$ such that
\begin{equation*}
\omega(\widehat{a\ast\initSeg{b}{L}})<|a|+L
\end{equation*}
and thus by an identical argument to the first case above (note that only Case 2 appeals to the assumption $(\ast)$) we have 
\begin{equation*}
\sbar\; v_1\; v_2\; v_3\; (\num{a}::\num{b}_0::\cdots::\num{b}_{L-1})\btl \dss{\sbar}(\omega,g,h,a\ast\initSeg{b}{L})
\end{equation*}
a contradiction. Thus $\sbar\; v_1\; v_2\; v_3\; \num{a}\btl \dss{\sbar}(\omega,g,h,a)$ for all $a\in\NN^\ast$.
\end{proof}

This establishes the requirement of Theorem \ref{thm-sound} for the functional symbol $\sbar$. A simple adaptation of Lemma \ref{lem-barterm} can be used to to show that the same condition holds for the auxiliary term $\sbar_1$ together with a suitably defined $\dss{\sbar_1}$, while a straightforward induction achieves the same for $\fold_\rho$, the details of which we also omit. As a result, Theorem \ref{thm-sound} yields the following:
\begin{thm}
For all closed terms $e:\rho$ in the bar recursive target language of Section \ref{sec-prog-spec} we have $e\downarrow v$ for some $v$. 
\end{thm}
The above theorem represents a normalization result for bar recursion, where for illustrative purpose we only consider the bar recursor of lowest type. Strong normalization of bar recursion (of arbitrary type) was first established by \cite{Vogel(1976.0)}, and then (without infinite terms) in \cite{Bezem(1985.1)}. Our approach, which also replaces infinite terms by appealing to the construction of choice sequences in the model, has been explored in much more detail by Berger \cite{Berger(2005.0),Berger(2006.0)}, where also normalization results for variants of bar recursion (including the so-called BBC functional \cite{BBC(1998.0)} and open recursion \cite{Berger(2004.0)}) are established. One crucial difference is that we work in a total model instead of a partial model, and it is open whether or not we can extend our framework (which works in total models) to incorporate other forms of recursion such as those mentioned above.

\begin{rem}
\label{rem-bars}
A well known result of Schwichtenberg \cite{Schwichtenberg(1979.0)} asserts that bar recursion of lowest type when applied to primitive recursive parameters is actually definable within System T (a generalisation of which has been more recently presented in \cite{OliStei(2017.0)}). Thus our bar recursive language as presented here would technically also be definable in System T, and as such a detour through the continuous functionals is not strictly necessary in order to prove termination. However, as already pointed out in Remark \ref{rem-bar}, we could readily extend our language to include bar recursion of arbitrary type, and the termination proof given as Lemma \ref{lem-barterm} would generalise accordingly. We also note that our approach is modular, and thus any extension of our bar recursive language with new function symbols which can be given a suitable interpretation is also terminating as a whole.
\end{rem}

\subsection{Majorizability}
\label{sec-simple-maj}

Our final simple application for the case $C=\{0\}$ will be an adaptation of Howard's majorizability relation \cite{Howard(1973.0)} to higher-order rewrite systems. Majorizability is an extension of the usual ordering $\leq$ on $\NN$ to functionals of arbitrary finite type. In addition to forming the basis for interesting models of higher-order calculi, majorizability plays an essential role in the \emph{proof mining} program, where it is used in conjunction with G\"{o}del's functional interpretation to form the \emph{monotone} functional interpretation (introduced in \cite{Kohlenbach(1996.0)}), which is crucial for interpreting forms of compactness related to the binary K\"{o}nig's lemma (for a detailed background on majorizability and its use in proof mining, see \cite{Kohlenbach(2008.0)}). For simplicity we work over a target language with a single base type $\nat$ with $D_\nat=\NN$. We define
\begin{equation*}
\num{n}\wtl_\nat m:\Leftrightarrow n\leq m. 
\end{equation*}
In this case, the logical relation $e\btl_\rho a$ is just a variant of the usual majorizability relation $\maj{\rho}$ at all finite type, which in an equational setting is defined as follows (cf. \cite{Howard(1973.0)} or \cite[Definition 3.34]{Kohlenbach(2008.0)})
\begin{equation*}
\begin{aligned}
x^\ast\maj{\nat} x&:\Leftrightarrow x^\ast\geq x \\
x^\ast\maj{\rho\to\tau} x&:\Leftrightarrow \forall y^\ast,y(y^\ast\maj{\rho} y\Rightarrow x^\ast y^\ast\maj{\tau} xy).
\end{aligned}
\end{equation*}
We compare this to our relation, which is based on the same idea but now phrased in terms of our programming language, and can be expanded as follows:
\begin{equation*}
\begin{aligned}
e\btl_\nat m&\Leftrightarrow (\exists n\in\NN)(e\downarrow \num{n}\wedge n\leq m)\\
e\btl_{\rho\to\tau} f&\Leftrightarrow (\exists u)(e\downarrow u\wedge (\forall v,a)(v\lhd_\rho a\Rightarrow (\exists w)(uv\downarrow w\wtl_\tau f(a)))). 
\end{aligned}
\end{equation*}
To illustrate this correspondence further, note that in an equational calculus we would have
\begin{equation*}
x^\ast \maj{\nat\to\nat} x:\Leftrightarrow \forall y^\ast,y(y^\ast\geq y\Rightarrow x^\ast y^\ast\geq xy)
\end{equation*}
whereas in our setting this would be rendered as
\begin{equation*}
e\btl_{\nat\to\nat} f\Leftrightarrow (\exists u)(e\downarrow u\wedge (\forall n,m)(n\leq m\Rightarrow (\exists k)(u\num{n}\downarrow \num{k}\wedge k\leq fm))).
\end{equation*}
In fact, we can prove the following key lemma, which is analogous to \cite[Lemma 3.35 (iii)]{Kohlenbach(2008.0)}:
\begin{lem}
\label{lem-maj}
Suppose that $e:\rho$ for $\rho:=\rho_1\to\ldots\to\rho_n\to\nat$. Then
\begin{equation*}
e\btl_\rho g\Leftrightarrow (\forall v_1,a_1,\ldots,v_n,a_n)\left(\bigwedge_{i=1}^n v_i\wtl a_i\Rightarrow (\exists k)( ev_1\ldots v_n\downarrow \num{k}\wedge k\leq g(a_1,\ldots,a_n))\right)
\end{equation*}
\end{lem}
\begin{proof}
Induction on $n$. For $n=0$ this is just the definition of $\btl_\nat$, and for the induction step we have
\begin{equation*}
\begin{aligned}
&e\btl g\\
\Leftrightarrow &(\exists u)(e\downarrow u\wedge (\forall v_0,a_0)(v_0\wtl a_0\Rightarrow uv_0\btl g(a_0))\\
\Leftrightarrow &(\exists u)(e\downarrow u\wedge (\forall v_0,a_0,\ldots v_n,a_n)(\bigwedge_{i=0}^n v_i\wtl a_i\Rightarrow (\exists k)(uv_0\ldots v_n\downarrow \num{k}\wedge k\leq g(a_0,\ldots a_n)))\\
\Leftrightarrow & (\forall v_0,a_0,\ldots v_n,a_n)(\bigwedge_{i=0}^n v_i\wtl a_i\Rightarrow (\exists k)(ev_0\ldots v_n\downarrow \num{k}\wedge k\leq g(a_0,\ldots a_n))
\end{aligned}
\end{equation*}
where for the last step we use $ev\downarrow w\Leftrightarrow (\exists u)(e\downarrow u\wedge uv\downarrow w)$, which follows directly from the operational semantics of the language.
\end{proof}
The following lemma is analogous to \cite[Lemma 3.66]{Kohlenbach(2008.0)}.
\begin{lem}
\label{lem-maj0}
Suppose that $e:\nat\to\rho$ for $\rho:=\rho_1\to\ldots\to\rho_k\to\nat$ and that $e\; \num{n}\btl_\rho g(n)$ for all $n\in\NN$. Then $e\btl_{\nat\to\rho} g^M$ for 
\begin{equation*}
g^M(m,a_1,\ldots,a_k):=\max\{g(i,a_1,\ldots,a_k)\; | \; i\leq m\}
\end{equation*}
\end{lem}
\begin{proof}
By Lemma \ref{lem-maj} we have that $e\btl g^M$ is equivalent to
\begin{equation*}
(\forall m,n,v_1,a_1,\ldots,v_k,a_k)(m\leq n\wedge \bigwedge_{i=1}^k v_i\wtl a_i\Rightarrow (\exists l)(e\num{n}v_1\ldots v_k\downarrow \num{l}\wedge l\leq g^M(m,a_1,\ldots,a_k)))
\end{equation*}
and the result then follows by a straightforward induction on $m$.
\end{proof}
Now, supposing for example that we work in an instance of our language, as in Section \ref{sec-prog-syst}, where there are just two constructor symbols $\zero$ and $\suc$, which we interpret in the usual way as $\dss{\zero}:=0$ and $\dss{\suc}(n):=n+1$ so that they are compatible with $\wtl$. In this case our main soundness theorem becomes
\begin{thm}
\label{thm-maj}
If for all function symbols $f:\rho_1\to\ldots\to \rho_n\to\rho$ of our language there is a suitably defined $\dss{f}$ satisfying $fv_1\ldots v_n\btl_\rho \dss{f}a_1\ldots a_n$ whenever $v_i\wtl_{\rho_i}a_i$ for all $i=1,\ldots,n$ and values $v_i$, then any term $t:\rho$ has an interpretation $\ds{\mt{t}}\in \ds{\rho}$ which majorizes it, in the sense that $\sigma\wtl \xi$ implies $t\sigma\btl_\rho \ds{\mt{t}}\xi$. In particular, whenever $t:\nat$ and $\sigma\wtl\xi$, we have $e\sigma\downarrow \num{n}$ with $n\leq \ds{\mt{t}}\xi$.
\end{thm}
\begin{cor}
\label{cor-Tmaj}
Any term $t:\rho$ in System T has an interpretation $\ds{\mt{t}}$ which majorizes it.
\end{cor}

\begin{proof}
We just need to interpret the recursor, to which end we define $R:\ds{\rho}\Rightarrow (\NN\Rightarrow\ds{\rho}\Rightarrow\ds{\rho})\Rightarrow\NN\Rightarrow\ds{\rho}$ by
\begin{equation*}
R(a,f,n)=\begin{cases}a & \mbox{if $n=0$}\\ f(n',R(a,f,n')) & \mbox{if $n=n'+1$}.\end{cases}
\end{equation*}
then it's easy to show that whenever $v_1\wtl_\rho a$, $v_2\wtl_\rho f$ and $n\in\NN$ then 
\begin{equation*}
\rec\; v_1\; v_2\; \num{n}\btl_\rho R(a,f,n)
\end{equation*}
using induction on $n$. Then by Lemma \ref{lem-maj0} it follows that 
\begin{equation*}
\rec\; v_1\; v_2\btl_{\nat\to\rho} R(a,f)^M,
\end{equation*}
and thus defining $\dss{\rec}(a,f,n)=R(a,f)^M(n)$ it follows that for $n\leq m$ we have $\rec\; v_1\; v_2\; \num{n}\btl_\rho\dss{\rec}(a,f,m)$.
\end{proof}
Corollary \ref{cor-Tmaj} is closely related to classic results of Howard \cite{Howard(1973.0)} and Bezem \cite{Bezem(1985.0)} which show that various type structures of majorizable functionals are a model of System T. This also holds more generally for bar recursion, and we conjecture that a variant of Theorem \ref{thm-maj} dealing with $\sbar$ as defined in Section \ref{sec-prog-spec} can be proven within our framework, though we do not attempt to work out the details here. The most important property of our Theorem \ref{thm-maj} is that it can be extended in a uniform way to a number of different programming languages.

\section{Extracting moduli of continuity for functional languages}
\label{sec-cont}

We now utilise the monadic part of our soundness theorem for the first time, to present a simple framework in which we can extract moduli of continuity from those functionals of type level $2$ which are definable by a term in our target language.

That all type level $2$ functionals definable in System T have a modulus of continuity also definable in T was apparently first presented by Kreisel in lectures from 1971/72, and early proofs can be found in e.g.~\cite{Schwichtenberg(1977.0)} and \cite{Troelstra(1973.0)} (Theorem 2.3.9). The approach presented here is more closely connected to recent work, where continuity and the extraction of corresponding moduli are established by appealing to monads and effects (e.g.~\cite{CoqJab(2012.0),Escardo(2013.0),Hernest(2007.1),Kawai(2019.0),BicRah(2018.0),Xu(2020.0)}, though this list is by no means exhaustive). A related monadic translation in all finite types due to the present author, but using the state monad and working in an equational setting, is also the basis for \cite{Powell(2018.1)}. 

Our main motivation in studying continuity lies not in the fact that this result is new or surprising in itself, but in demonstrating that it fits elegantly into our uniform framework, which can, moreover, be readily extended to other languages and forms of recursion. 

\subsection{The languages $\bP$, $\bP^T$ and $\bP_g$}
\label{sec-cont-bP}

We illustrate our approach by working in a target language with a single datatype $\nat$ and the usual constructors $\zero$ and $\suc$, and assume that we have fixed some collection of function symbols $\func$. We call this base language $\bP$, and observe that our simple variant of System T outlined in Section \ref{sec-prog-syst} is an instance of $\bP$, which we later label as $\bP^T$. 

Now, to an arbitrary function $g:\NN^\NN$ we associate a variant $\bP_g$ (resp. $\bP_g^T$) of our base language, whose function symbols are the same as those of $\bP$ but now extended to include an oracle $\alpha:\nat\to\nat$ whose defining equations are given by
\begin{equation*}
\alpha \; \num{n}\leadsto \num{g(n)}
\end{equation*}
for each $n\in\NN$, where here $\num{g(n)}$ denotes the numeral representation of $g(n)$. Note that for any two functions $g,h:\NN^\NN$ the languages $\bP_g$ and $\bP_h$ have the same terms, and so differ only in the operational semantics of $\alpha$. We clearly distinguish the operational semantics of $\bP$ from that of $\bP_g$ by using $e\downarrow_g v$ to denote the big-step relation of the latter, which also includes the derived rule $\alpha\; \num{n}\downarrow_g \num{g(n)}$.

\subsection{The monadic translation}
\label{sec-cont-mon}

We now define $D_\nat:=\NN$ as usual, but this time set $C:=\NN^\ast$. Fixing some $g:\NN^\NN$, we define the logical relation $\btl_g$ on terms of $\bP_g$ via $\num{n}\wtl m:\Leftrightarrow n=m$ and
\begin{equation*}
e\Downarrow^c_g v:\Leftrightarrow (e\downarrow_g v\wedge \forall h(g=_c h\Rightarrow e\downarrow_h v))
\end{equation*}
where 
\begin{equation*}
g=_c h:\Leftrightarrow(\forall i\in c)(g(i)=h(i))
\end{equation*}
Finally, our denotational semantics $\ds{\cdot}_g$ of the monadic metalanguage arising from $\bP_g$, which is also parametrised by $g$, is defined by firstly instantiating our monadic components as
\begin{equation*}
\epsilon:=[] \ \ \ \ \ \ \inc(c):=c \ \ \ \ \ \ \comp(c_0,c_1,c_2):=c_0\ast c_1\ast c_2.
\end{equation*}
and interpreting our constructors in the obvious way. We assume that for each function symbol we have an interpretation $\dss{f}$ which is independent of $g$, and as such, this parameter only plays a role in the interpretation of the oracle function symbol $\alpha:\nat\to\nat$, whereby we set $\dss{\alpha}_g:\NN\Rightarrow\NN^\ast\times\NN$ to be
\begin{equation*}
\dss{\alpha}_g(n):=([n],g(n)).
\end{equation*}
If we are working in either the full set theoretic model $\modset$ or a continuous model $\modcont$, everything up until now is well defined. As before, it is clear that the interpretation of the constructors is compatible with $\wtl_\nat$, and we also have the following:
\begin{lem}
The tuple $(\epsilon,\inc,\comp)$ is compatible with $\Downarrow_g$.
\end{lem}

\begin{proof}
We deal with each rule in turn. We have $v\downarrow_h v$ for any values and parameter $h$, and so $v\Downarrow^c_g v$ is true for any $c$, and in particular $c=[]$.

If $r[u/x]\Downarrow^c_g v$ then $r[u/x]\downarrow_g v$ and whenever $g=_c h$ then $r[u/x]\downarrow_h v$. But then $(\lambda x.r)u\downarrow_g v$ in $\bP_g$, and since $g=_{\inc(c)}h$ is just $g=_ch$ we have $r[u/x]\downarrow_h v$, and so also $(\lambda x.r)u\downarrow_h v$ in $\bP_h$.

Finally, if $e\Downarrow^{c_0}_{\rho\to\tau,g} u$, $e'\Downarrow^{c_1}_{\rho,g} v$ and $uv\Downarrow^{c_2}_{\tau,g} w$, then we have $e\downarrow_g u$, $e'\downarrow_g v$ and $uv\downarrow_g w$ and thus $ee'\downarrow_g w$. Now suppose that $g=_{c_0\ast c_1\ast c_2} h$. Then we have $g=_{c_i} h$ for $i=0,1,2$, and thus $e\downarrow_h u$, $e'\downarrow_h v$ and $uv\downarrow_h w$ and therefore $ee'\downarrow_h w$. This establishes $ee'\Downarrow^{c_0\ast c_1\ast c_2}_\tau w$.
\end{proof}
Our main soundness theorem gives rise in this case to the following continuity theorem.
\begin{thm}
\label{thm-cont-main}
Suppose that for all non-oracle function symbols of $f:\rho_1\to\ldots\to\rho_n\to\rho$ in the base language $\bP$ and any $g:\NN^\NN$ we have $fv_1\ldots v_k\btl_{\rho,g} \dss{f} a_1\ldots a_k$ whenever $v_i\wtl_{\rho_i,g} a_i$. Let $e:(\nat\to\nat)\to\nat$ be a closed term of $\bP$. Then we have
\begin{equation*}
(\forall g,h\in\NN^\NN)(e\alpha\downarrow_g \num{n}\wedge (g=_ch\Rightarrow e\alpha\downarrow_h \num{n})) 
\end{equation*}
for $(c,n):=\ds{\mt{e\alpha}}_g$.
\end{thm}

\begin{proof}
By assumption the translation $\ds{\mt{\cdot}}_g$ is sound with respect to $\btl_g$ for the function symbols of $\bP$ (i.e. the non-oracle function symbols of $\bP_g$) so to verify soundness of the extended language $\bP_g$ it remains to check the oracle symbol $\alpha$. But for any $n\in\NN$ we have
\begin{equation*}
\alpha\; \num{n}\downarrow_g \num{g(n)}\wedge (g=_{[n]} h\Rightarrow \alpha\; \num{n}\downarrow_h \num{h(n)}=\num{g(n)})
\end{equation*}
and thus $\alpha\; \num{n}\btl_{\nat,g} ([n],g(n))=\dss{\alpha}_g(n)$. Therefore it follows that for any closed $e':\rho$ in $\bP_g$ we have $e'\btl_{\rho,g}\ds{\mt{e'}}_g$. In particular, for $e:(\nat\to\nat)\to\nat$, setting $e':=e\alpha:\nat$ we obtain $e\alpha\btl_{\nat,g}(c,n)$ for $(c,n):=\ds{\mt{e\alpha}}_g$. But this just means that $e\alpha\Downarrow^c_g\num{n}$, which is exactly what we want to show. 
\end{proof}

\begin{cor}
\label{cor-cont-T}
Theorem \ref{thm-cont-main} holds for $\bP^T$ (i.e. System T as defined in Section \ref{sec-prog-syst}).
\end{cor}

\begin{proof}
We need to find a suitable interpretation of the recursors $\rec_\rho$ independent of $g$, to which end we define $\dss{\rec}_\rho$ by
\begin{equation*}
\dss{\rec}(a,f,n)=\begin{cases}([],a) & \mbox{if $n=0$}\\ fn'\circ \dss{\rec}(a,f,n') & \mbox{if $n=n'+1$} \end{cases}
\end{equation*}
where $\circ$ is as defined in Lemma \ref{lem-compapp}.
Then using Lemma \ref{lem-compapp} together with induction it is straightforward to show that for any parameter $g$, whenever $v_1\wtl_{\rho,g} a$ and $v_2\wtl_{\nat\to\rho\to\rho,g} f$ then $\rec\; v_1\; v_2\; \num{n}\btl_{\rho,g} \dss\rec(a,f,n)$. 

To see this, note that for the base case we have $\rec\; v_1\; v_2\; \zero\downarrow_g v_1\wtl_{\rho,g} a$, and since for any $h$ it's also the case that $\rec\; v_1\; v_2\; \zero\downarrow_h v_1$ this implies that $\rec\; v_1\; v_2\; \zero\Downarrow_g^{[]} v_1$ and thus $\rec\; v_1\; v_2\; \zero\btl_{\rho,g}([],a)$.

For the induction step, by $v_2\wtl_{\nat\to\rho\to\rho,g}f$ we have $v_2\; \num{n}\btl_{\rho\to\rho,g} fn$ and by the induction hypothesis $\rec\; v_1\; v_2\; \num{n}\btl_{\rho,g}\dss{\rec}(a,f,n)$. Therefore by Lemma \ref{lem-compapp} we have
\begin{equation*}
v_2\; \num{n}\; (\rec\; v_1\; v_2\; \num{n})\btl_{\rho,g} fn\circ \dss{\rec}(a,f,n)=\dss{\rec}(a,f,n+1)
\end{equation*}
which implies that $v_2\; \num{n}\; (\rec\; v_1\; v_2\; \num{n})\Downarrow_g^c w\wtl_{\rho,g} r$ for $(c,r):=\dss{\rec}(a,f,n+1)$. But since by our operational semantics $v_2\; \num{n}\; (\rec\; v_1\; v_2\; \num{n})\downarrow_h w$ implies $\rec\; v_1\; v_2\; \suc(\num{n})\downarrow_h w$ for any $h$, it follows that $\rec\; v_1\; v_2\; \suc(\num{n})\Downarrow^c_g w$ and hence $\rec\; v_1\; v_2\; \suc(\num{n})\btl_{\rho,g} (c,r)$.
\end{proof}

\subsection{Continuity of functionals $\NN^\NN\to\NN$ definable in System T}
\label{sec-cont-T}

We now use the main theorem above to prove a more traditional formulation of the continuity of System T functionals, namely that any set-theoretic object $F:\NN^\NN\to\NN$ which is definable in System T is continuous in the following sense:
\begin{equation*}
\forall g\exists N\forall h(g=_N h\Rightarrow F(f)=F(g)),
\end{equation*}
where we write $g=_N h$ if $g(i)=h(i)$ for all $i<N$. We first set up a connection between $\bP^T$ and the usual formulation of System T as subset of the set theoretic type structure $\modset$. 
\begin{defi}
\label{def-semantic}
For each $\Gamma\vdash t:\rho$ in our target language $\bP^T$ we define the usual semantic interpretation $[\Gamma]\vdash [t]:[\rho]$ in $\modset$ as follows: $[\nat]:=\NN$ and $[\rho\to\tau]:=[\rho]\Rightarrow[\tau]$, and on terms:
\begin{gather*}
[x]\eta:=\eta(x) \ \ \ \ \ \ [\lambda x.t]\eta:=\mlambda a.[t]\eta\{x\mapsto a\} \ \ \ \ \ \ [ts]\eta:=[t]\eta([s]\eta)\\
[\zero]:=0 \ \ \ \ \ \ [\suc](n):=n+1 \\
[\rec_\rho](a,h,0):=a \ \ \ \ \ \ [\rec_\rho](a,h,n+1)=hn([\rec](a,h,n))
\end{gather*}
\end{defi}
This interpretation then eliminates any difference between our formulation of System T as a call-by-value language and those based on equational calculi: An object of $\modset$ is definable in System T precisely when it is of the form $[e]$ for some closed term $e$ in $\bP^T$. We now extend the above definition to incorporate our oracle symbol.
\begin{defi}
For each $g\in\NN^\NN$ and $\Gamma\vdash t:\rho$ in $\bP^T_g$, we define $[\Gamma]_g\vdash [t]_g:[\rho]$ by extending the clauses of Definition \ref{def-semantic} with
\begin{equation*}
[\alpha]_g(n):=g(n).
\end{equation*}
It is easy to prove that if $t$ is a term of $\bP^T_g$ which does not contain $\alpha$, then it is also a term of $\bP^T$ and moreover $[t]_g\eta=[t]\eta$ for any $\eta$.
\end{defi}
We now need a result which confirms that the pure interpretation $[\cdot]_g$ can be related to the semantic part of $\ds{\mt{\cdot}}_g$.
\begin{lem}
\label{lem-logical}
Define the auxiliary logical relation $\sim_\rho$ on $[\rho]\times \ds{\mt{\rho}}$ as follows:
\begin{equation*}
m\sim_\nat n:\Leftrightarrow m=_\NN n \ \ \ \ \ \ f\sim_{\rho\to\tau} g:\Leftrightarrow \forall a,b(a\sim_\rho b\Rightarrow fa\sim_\tau \pi_1(gb))
\end{equation*}
We write $\eta\sim_\Gamma \xi$ if $\eta(x)\sim_{\rho}\xi(x)$ for all $x:\rho\in \Gamma$. Then for any $g\in\NN^\NN$ and $\Gamma\vdash t$ in $\bP^T_g$, if $\eta\sim_\Gamma\xi$ then we have $[t]_g\eta\sim_\rho \pi_1(\ds{\mt{t}}_g\xi)$.
\end{lem}

\begin{proof}
A simple induction on terms. For variables this follows directly. For abstraction, if $a\sim_\rho b$ then 
\begin{equation*}
[\lambda x.t]_g\eta(a)=[t]_g\eta\{x\mapsto a\}\sim_\tau \pi_1(\ds{\mt{t}}_g\xi\{x\mapsto b\})=\pi_1(\pi_1(\ds{\mt{\lambda x.t}}_g\xi(b))) 
\end{equation*}
and thus $[t]_g\eta\sim_{\rho\to\tau} \pi_1(\ds{\mt{\lambda x.t}}_g\xi)$. For application, since $[t]_g\eta\sim_{\rho\to\tau} \pi_1(\ds{\mt{t}}_g\xi)$ and $[s]_g\eta\sim_{\rho} \pi_1(\ds{\mt{s}}_g\xi)$ we have $[ts]_g\eta\sim_\rho \pi_1(\pi_1(\ds{\mt{t}}_g\xi)(\pi_1(\ds{\mt{s}}_g\xi)))=\pi_1(\ds{\mt{ts}}_g\xi)$. For the constants the result is trivial, and for the recursor we must show that
\begin{equation*}
a\sim_\rho a'\wedge h\sim_{\nat\to\rho\to\rho} h'\to [\rec_\rho](a,h,n)\sim_\rho \pi_1(\dss{\rec}_\rho(a',h',n))
\end{equation*}
for all $n\in\NN$, which is a simple induction appealing to a variant of the argument used for application. Finally, for the oracle symbol, we must show that
\begin{equation*}
[\alpha]_g(n)=\pi_1(\dss{\alpha}_g(n))=g(n)
\end{equation*}
for any $n\in\NN$, which is the case by definition.
\end{proof}

\begin{cor}
\label{cor-ad}
Suppose $e:(\nat\to\nat)\to\nat$ is a closed term of $\bP^T$. Then for any $g\in\NN^\NN$ we have $\pi_1(\ds{\mt{e\alpha}}_g)=[e](g)$.
\end{cor}

\begin{proof}
If $e:(\nat\to\nat)\to\nat$ is a closed term of $\bP^T$ then in particular it can be viewed as a term of $\bP^T_g$ which does not contain $\alpha$, and as such we have $[e]_g=[e]$. Since $[\alpha]_g=g$ by definition, by Lemma \ref{lem-logical} we have $[e\alpha]_g= \pi_1(\ds{\mt{e\alpha}}_g)$ and thus
\begin{equation*}
[e](g)=[e]_g([\alpha]_g)=[e\alpha]_g=\pi_1(\ds{\mt{e\alpha}}_g)
\end{equation*}
and we are done.
\end{proof}

\begin{thm}
\label{thm-cont}
Suppose that $F:\NN^\NN\to\NN$ is definable in System T. Then
\begin{equation*}
\forall g,h(g=_{\Phi(g)} h\Rightarrow F(g)=F(h))
\end{equation*}
where $\Phi$ is defined by
\begin{equation*}
\Phi(g):=\max\{ j\; : \; j\in\pi_0(\ds{\mt{e}}\circ([],\mlambda n.([n],gn)))\}+1
\end{equation*}
for some closed term $e$ representing $F$.
\end{thm}

\begin{proof}
If $F$ is definable in System T, there is some closed term $e:(\nat\to\nat)\to\nat$ of $\bP^T$ such that for all $f\in\NN^\NN$ we have $[e](f)=F(f)$. Observing that 
\begin{equation*}
\ds{\mt{e\alpha}}_g=\ds{\mt{e}}\circ\ds{\mt{\alpha}}_g=\ds{\mt{e}}\circ([],\mlambda n.([n],gn))
\end{equation*}
we have $g=_{\Phi(g)} h$ implies that $g=_{\pi_0(\ds{\mt{e\alpha}}_g)} h$. Thus by Theorem \ref{thm-cont-main}, which by Corollary \ref{cor-cont-T} applies to $\bP^T$, we have $e\alpha\downarrow_g\num{n}$ and $e\alpha\downarrow_h \num{n}$ for $n=\pi_1(\ds{\mt{e\alpha}}_g)$, and by Corollary \ref{cor-ad} we have $\pi_1(\ds{\mt{e\alpha}}_g)=[e](g)=F(g)$. Now, by Theorem \ref{thm-cont-main} again, we have $e\alpha\downarrow_h\num{m}$ where $m=\pi_1(\ds{\mt{e\alpha}}_h)$. But by uniqueness of derivations in $\bP^T_h$, we therefore have $\num{n}=\num{m}$ and therefore $n=m=\pi_1(\ds{\mt{e\alpha}}_h)=[e](h)=F(h)$ and thus $F(g)=F(h)$. 
\end{proof}

We also note that the term $\ds{\mt{e}}$ is, in turn, constructed using only basic operations of the lambda calculus, cartesian product and list operations, together with primitive recursion, and can thus also be defined in some suitable variant of System T.

This section demonstrates how our main result on continuity (Theorem \ref{thm-cont-main}) can be reformulated in terms of continuity of functionals $\NN^\NN\to\NN$. It goes without saying that our approach throughout this section could be extended to other systems, such as extensions of System T with bar recursion, though we don't go through any of the details here.

\section{A denotational cost semantics}
\label{sec-cost}

In our final application, we show how various denotational complexity semantics of functional languages can be re-obtained in a uniform way in our setting. As mentioned in the introduction, the static cost analysis of functional programs dates back to at least the 1980s \cite{Sands(1990.0),Shultis(1985.0)}, and has been explored more recently in a generalised categorical setting \cite{VanStone(2003.0)} and from the perspective of automating the analysis of synthesised programs \cite{Benzinger(2004.0)}, though naturally these papers constitute just a few representative examples of the very diverse literature on complexity in higher types.

A distinguishing feature shared by each of the works mentioned above is that they view the complexity of a higher-order functional as a higher-order object in its own right, namely a cost functional which not only tracks the number of steps it takes for a higher-order term $t$ to reduce to a normal form $\lambda x.u$, but encodes information about the cost of evaluating $(\lambda x.u)v$ for any input $v$, and so on. The synthesis of cost functionals from higher-order programs in this sense is readily accomplished in our setting by instantiating the monadic part of our translation as a simple step counting operation. In particular, the result of our translation is an object $\ds{\mt{e}}$ which is always given as recursive equation in our model which reflects the syntactic structure of the original term $e$. This recursive equation can then be solved to give a closed form expression for the cost of evaluating $e$ (see Example \ref{ex-search} for an illustration of this in the case of Spector's search functional).

However, we also show that different choices of our \emph{semantic} component lead to interesting characterisations of complexity. In particular, by combining the aforementioned cost operation with a form of majorizability similar to that discussed in Section \ref{sec-simple-maj}, we obtain in a uniform way the bounded cost semantics explored recently by Danner et al.~in \cite{DanLicRam(2015.0),DanPayRoy(2013.0)}.

On top of all this, we apply our translation not just to variants of System T, as is often the case in the literature, but to bar recursive extensions, for which soundness of our cost analysis requires us to appeal to semantic continuity principles as in Section \ref{sec-simple-ds}.

\subsection{Exact cost expressions for call-by-value languages}
\label{sec-cost-exact}

We begin by defining what we mean by the \emph{cost} of a closed term $e$ in our parametrised target language. We do this using annotated big step relations $e\downarrow^c v$ for $c\in\NN$, where intuitively, $e\downarrow^c v$ iff $e\downarrow v$ in $c$ rewrite steps. We define this formally in Figure \ref{fig-cost}.
\begin{figure}[t]
\begin{gather*}
v\downarrow^0_\rho v\mbox{ \ \ for values $v$} \\[2mm] 
\frac{r[u/x]\downarrow^c_\tau v}{(\lambda x.r)u\downarrow^{c+1}_\tau v}\mbox{ \ \ for $x:\rho\vdash r$ and $u:\rho$ a value} \\[2mm]
\frac{r\sigma\downarrow^c_\rho v}{fv_1\ldots v_n\downarrow^{c+1}_\rho v}\mbox{ \ \ for $v_1,\ldots,v_n=p_1\sigma,\ldots,p_n\sigma$ and $fp_1\ldots p_n\leadsto r$}\\[2mm]
 \frac{e\downarrow^{c_0}_{\rho\to\tau} u \ \ \ \ \ e'\downarrow^{c_1}_\rho v \ \ \ \ \ \ uv\downarrow^{c_2}_\tau w}{ee'\downarrow^{c_0+c_1+c_2}_\tau w}\mbox{ \ \ \ if one of $e$, $e'$ is not a value}
\end{gather*}
\caption{Cost semantics of target language}
\label{fig-cost}
\hrulefill
\end{figure}

Throughout this section our main illustrative example will be our list based variant of System T from Section \ref{sec-prog-lyst}, together with its bar recursive extension from Section \ref{sec-prog-spec}. As such, we formulate our soundness theorem (Theorem \ref{thm-exact}) in terms of an arbitrary target language over datatypes $\nat$ and $\nat^\ast$ with constructors $\zero$, $\suc$ $\emp$ and $\cons$, noting that both of the aforementioned languages follow as simple instances of this. Our monadic translation for exact costs is based on setting $C:=\NN$ with $D_\nat:=\NN$ and $D_{\nat^\ast}:=\NN^\ast$, and we define our logical relation via
\begin{equation*}
e\Downarrow^c_\rho v:\Leftrightarrow e\downarrow^c_\rho v \ \ \ \ \ \ \num{n}\wtl_\nat m:\Leftrightarrow n=m \ \ \ \ \ \ \num{a}\wtl_{\nat^\ast} b:\Leftrightarrow a=b
\end{equation*}
and thus $\btl_\rho$ becomes
\begin{equation*}
\begin{aligned}
e\btl_\rho (c,a)&:\Leftrightarrow (\exists v)(e\downarrow^c_\rho v\wedge v\wtl_\rho a)\\
u\wtl_{\rho\to\tau} f&:\Leftrightarrow (\forall v,a)(v\wtl_\rho a\Rightarrow uv\btl_\tau f(a)).
\end{aligned}
\end{equation*}
For our denotational semantics of terms, we define the monadic components by
\begin{equation*}
\epsilon:=0 \ \ \ \ \ \ \inc(c):=c+1 \ \ \ \ \ \ \comp(c_0,c_1,c_2):=c_0+c_1+c_2
\end{equation*}
and interpret our constructors in the obvious way, so that the interpretations are in particular compatible with $\wtl$. Note that it is clear from the definition of cost that $(\epsilon,\inc,\comp)$ is compatible with $\Downarrow$. Thus our main soundness proof (Theorem \ref{thm-sound}) results in the following characterisation of costs in higher types.
\begin{thm}
\label{thm-exact}
Suppose that for each function symbol $f$ in our target language there is a suitable interpretation $\dss{f}$ such that $fv_1\ldots v_k\btl \dss{f}a_1\ldots a_k$ whenever $v_i\wtl_{\rho_i}a_i$. Then for any closed $e:\rho$ we have
\begin{equation*}
e\downarrow^{\pi_0\ds{\mt{e}}}v
\end{equation*}
for some value $v:\rho$.
\end{thm}
It is relatively straightforward to apply the metatheorem above to the full list based variant of System T, and even its extension with bar recursion. For our basic operators, we have e.g.~$\num{n}+\num{m}\btl_{\nat}(1,n+m)$ and $\ext\; \num{a}\; \num{n}\btl_{\nat}(1,\hat{a}_n)$, and so on.
\begin{defi}
For $c\in \NN$ and $a\in \NN\times\rho$ define $c\dplus a:=(c+\pi_0a,\pi_1a)\in\NN\times\rho$.
\end{defi}
\begin{lem}
\label{lem-listT-cost}
\begin{enumerate}[(a)]

\item We have $\fold\; v_1\; v_2\;\num{a}\btl\dss{\fold}(b,h,a)$ for any $v_1\wtl b$, $v_2\wtl h$ and $a\in\NN^\ast$, where $\dss{\fold}$ is defined by 
\begin{equation*}
\dss{\fold}(b,h,a):=1\dplus \begin{cases}(0,b) & \mbox{if $a=[]$}\\ hn\circ\dss{\fold}(b,h,a') & \mbox{if $a=a'\ast n$}.\end{cases}
\end{equation*}

\item We have $\sbar\; v_1\; v_2\; v_3\; \num{a}\btl \dss{\sbar}(\omega,g,h,a)$ for any $v_1\wtl \omega$, $v_2\wtl g$, $v_3\wtl h$ and $a\in\NN^\ast$. where $\dss{\sbar}$ is defined by
\begin{equation*}
\begin{aligned}
\dss{\sbar}(\omega,g,h,a):=&4+\omega_0(\mlambda i.(1,\hat{a}_i))\\
&\dplus\begin{cases}g(a) & \mbox{if $\omega_1(\mlambda i.(1,\hat{a}_i))<|a|$}\\ ha\circ (0,\mlambda n\; . \; 1\dplus \dss{\sbar}(\omega,g,h,a\ast n)) & \mbox{otherwise}\end{cases}
\end{aligned}
\end{equation*} 
for $\omega_i(f):=\pi_i(\omega(f))$.
\end{enumerate}
\end{lem}

\begin{proof}
Part (a) is a simple induction on the length of $a$, observing that the derived rules for $\fold$ give rise to the following cost rules:
\begin{equation*}
\fold_\rho\; v_1\; v_2\; \varepsilon\downarrow^1_\rho v_1 \ \ \ \ \ \ \frac{v_2\; \num{n}\downarrow^{c_0}_{\rho\to\rho} u \ \ \ \ \ \ \fold_\rho\; v_1\; v_2\; \num{a}\downarrow^{c_1}_\rho v \ \ \ \ \ \ uv\downarrow^{c_2}_\rho w}{\fold_\rho\; v_1\; v_2\; \num{a}::\num{n}\downarrow^{1+c_0+c_1+c_2}_\rho w}
\end{equation*}
For part (b), we work in the continuous model $\modcont$ as in Section \ref{sec-simple-ds} and follows the same strategy as the proof of Lemma \ref{lem-barterm}. By keeping track of costs we see that that the derived operational semantics of bar recursion give rise to the following derived cost rules:
\begin{gather*}
\frac{v_1(\ext\; \num{a})\downarrow^{c_0}_\nat \num{k} \ \ \ \ \ \ k<|a| \ \ \ \ \ \ v_2\; \num{a}\downarrow^{c_1}_\nat \num{n}}{\sbar\; v_1\; v_2\; v_3\; \num{a}\downarrow^{4+c_0+c_1} \num{n}}\\[2mm]
\frac{v_1(\ext\; \num{a})\downarrow^{c_0}_\nat \num{k} \ \ \ \ \ \ k\geq |a| \ \ \ \ \ \ v_3\; \num{a}\; (\lambda x\; . \; \sbar\; v_1\; v_2\; v_3\; (\num{a}::x))\downarrow^{c_1}_\nat \num{n}}{\sbar\; v_1\; v_2\; v_3\; \num{a}\downarrow^{4+c_0+c_1}_\nat \num{n}}
\end{gather*}
To formally derive each of these rules involves a careful analysis of the defining rules of the bar recursor constants, along with the cost semantics as set out in Figure \ref{fig-cost}. For the first, note that from the defining equation
\begin{equation*}
\sbar_1\; f\; g\; h\; xs\; \zero\leadsto g\; xs
\end{equation*}
we can infer $\sbar_1\; v_1\; v_2\; v_3\; \num{a}\; \zero\downarrow^{1+c_1}\num{n}$ from $v_2\; \num{a}\downarrow^{c_1}\num{n}$. Now if $k<|a|$ it follows that $(\num{k}<\len\; \num{a})\downarrow^2\zero$, since this involves the reduction of two elementary operations, and thus $(v_1(\ext\; \num{a})<\len\; \num{a})\downarrow^{2+c_0}\zero$ follows from $v_1(\ext\; \num{a})\downarrow^{c_0}_\nat \num{k}$. Putting this all together we obtain
\begin{equation*}
\sbar_1\; v_1\; v_2\; v_2\; \num{a}\; (v_1(\ext\; \num{a})<\len\; \num{a})\downarrow^{3+c_0+c_1}\num{n}
\end{equation*}
and from the defining equation 
\begin{equation*}
\sbar\; f\; g\; h\; xs\leadsto \sbar_1\; f\; g\; h\; xs\; (f(\ext\; xs)<\len\; xs)
\end{equation*}
we derive $\sbar\; v_1\; v_2\; v_3\; \num{a}\downarrow^{4+c_0+c_1}\num{n}$. The second derived rule is established in a similar manner. Now, continuing with the main proof, we first note that $\ext\; \num{a}\wtl_{\nat\to\nat}\mlambda i.(1,\hat{a}_i)$, and thus since $v_1\wtl \omega$ it follows that $v_1(\ext\; \num{a})\btl_\nat \omega(\mlambda i.(1,\hat{a}_i))$ and thus $v_1(\ext\; \num{a})\downarrow^{c_0}_\nat \num{k}$ with 
\begin{equation*}
(c_0,k):=\omega(\mlambda i.(1,\hat{a}_i)).
\end{equation*}
As in Lemma \ref{lem-barterm}, we now assume inductively that 
\begin{equation*}
\sbar\; v_1\; v_2\; v_3\; \num{a}::\num{n}\btl \dss{\sbar}(\omega,g,h,a\ast n)
\tag{$\ast$}
\end{equation*}
for some fixed $a\in\NN^\ast$ and all $n\in\NN$, and seek to establish $\sbar\; v_1\; v_1\; v_2\; \num{a}\btl \dss{\sbar}(\omega,g,h,a)$. There are two cases to deal with.
\begin{itemize}

\item Case 1: $k<|a|$. Observing that from $v_2\wtl g$ we have $v_2\; \num{a}\downarrow^{c_1} \num{n}$ for $(c_1,n)=g(a)$, it follows that $\sbar\; v_1\; v_2\; v_3\; \num{a}\downarrow^{4+c_0+c_1} \num{n}$ where $(4+c_0+c_1,n)=(4+c_0)\dplus g(a)$.

\item Case 2: $k\geq |a|$. By our assumption $(\ast)$ we have
\begin{equation*}
(\sbar\; v_1\; v_2\; v_3\; \num{a}::x)[\num{n}/x]\downarrow^c\num{m}
\end{equation*}
for $(c,m):=\dss{\sbar}(\omega,g,h,a::n)$ and thus
\begin{equation*}
(\lambda x\; . \; \sbar\; v_1\; v_2\; v_3\; \num{a}::x)\num{n}\downarrow^{c+1}\num{m}
\end{equation*}
from which we can infer
\begin{equation*}
\lambda x\; . \; \sbar\; v_1\; v_2\; v_3\; \num{a}::x\btl_{\nat\to\nat} (0,\mlambda n\; . \; 1\dplus\dss{\sbar}(\omega,g,h,a::n)).
\end{equation*}
Since $v_3\wtl h$ we then have $v_3\; \num{a}\btl ha$ and thus by Lemma \ref{lem-compapp} it follows that 
\begin{equation*}
v_3\; \num{a}\; (\lambda x\; . \; \sbar\; v_1\; v_2\; v_3\; \num{a}::x)\downarrow^{c_1}\num{n}
\end{equation*}
for $(c_1,n)= ha\circ  (0,\mlambda n\; . \; 1\dplus\dss{\sbar}(\omega,g,h,a::n) )$. Therefore by the second derived rule we have $\sbar\; v_1\; v_2\; v_3\; \num{a}\downarrow^{4+c_0+c_1} \num{n}$ where $(4+c_0+c_1,n)=(4+c_0)\dplus ha\circ  (0,\mlambda n\; . \; 1\dplus\dss{\sbar}(\omega,g,h,a::n) )$. 
\end{itemize}
Putting both cases together we see that $\sbar\; v_1\; v_1\; v_2\; \num{a}\btl \dss{\sbar}(\omega,g,h,a)$ holds whenever $(\ast)$ holds. We now suppose as in Lemma \ref{lem-barterm} that it is not the case that $\sbar\; v_1\; v_2\; v_3\; \num{a}\btl\dss{\sbar}(\omega,g,h,a)$, and construct an infinite sequence $b$ on the metalevel such that 
\begin{equation*}
\neg (\sbar\; v_1\; v_2\; v_3\; (\num{a}::\num{b}_0::\ldots ::\num{b}_{l-1})\btl \dss{\sbar}(\omega,g,h,a\ast\initSeg{b}{l}))
\end{equation*}
for all $l\in\NN$. This time we apply a continuity argument to $\omega_1:(\NN\to\NN\times\NN)\to\NN$ and the sequence $\beta:=\mlambda i.(1,(a\ast b)_i)$, by which there exists some $N$ such that 
\begin{equation*}
\omega_1(\initSeg{\beta}{|a|+N}\ast\gamma)=\omega_1(\beta)
\end{equation*}
for all $\gamma:\NN\to\NN\times\NN$. In particular setting $L:=\max\{N,\omega_1(\beta)+1-|a|\}$ we have
\begin{equation*}
\omega_1(\mlambda i.(1,(\widehat{a\ast\initSeg{b}{L}})_i)=\omega_1(\beta)<|a|+L.
\end{equation*}
Thus for input $a\ast\initSeg{b}{L}$ Case 1 applies and we have 
\begin{equation*}
\sbar\; v_1\; v_2\; v_3\; (\num{a}::\num{b_0}::\ldots ::\num{b}_{L-1})\btl \dss{\sbar}(\omega,g,h,a\ast\initSeg{b}{L}),
\end{equation*}
a contradiction. Therefore it must be the case that $\sbar \; v_1\; v_2\; v_3\; \num{a}\btl\dss{\sbar}(\omega,g,h,a)$ for all $a\in\NN^\ast$, which completes the proof.
\end{proof}

\begin{cor}
Let $e:\rho$ be a closed term definable in System T plus bar recursion (in the sense of Section \ref{sec-prog-spec}). Then $e\downarrow^{\pi_0\ds{\mt{e}}} v$ for some value $v$.
\end{cor}

\begin{exa}
\label{ex-search}
Consider the so-called Spector search functional $\Phi:((\NN\to\NN)\to\NN)\to (\NN\to\NN)\to\NN$ given by
\begin{equation*}
\Phi(\omega,\beta,a)=\begin{cases}0 & \mbox{if $\omega(\hat{a})<|a|$}\\ 1+\Phi(\omega,\beta,a::\beta(|a|)). & \mbox{otherwise}\end{cases}
\end{equation*}
This functional was introduced by Howard in \cite{Howard(1968.0)}, though he attributes it to Kreisel. The idea is that $\Phi(\omega,\beta,[])$ forms a bound on how far we need to look to find some $N$ such that $\omega(\widehat{\initSeg{\beta}{N}})<N$, confirming that such an $N$ can be computed using bar recursion.

The search functional can be defined in our target language via the closed term
\begin{equation*}
\spec:=\lambda x,y,z\; . \; \sbar\; x\; (\lambda x'.\zero)(\lambda z',p\; . \; \suc(p(y(\len\; z'))))\; z.
\end{equation*}
Then given $v_1\btl \omega$ and $v_2\btl g$, it is not too difficult to work out that 
\begin{equation*}
\spec\; v_1\; v_2\; \num{a}\btl_{\nat} \phi(\omega,g,a)
\end{equation*}
where $\phi$ is defined as
\begin{equation*}
\phi(\omega,g,a)=\dss{\sbar}(\omega,\mlambda b'.(1,0),\mlambda b\; . \; (1,\mlambda f\; . \; (4+g_0(|b|)+f_0(g_1(|b|)),1+f_1(g_1(|b|)))),a).
\end{equation*}
Unwinding the defining equations of $\dss{\sbar}$ we obtain
\begin{equation*}
\begin{aligned}
\phi(\omega,g,a)&=4+\omega_0(\mlambda i.(1,\hat{a}_i))\dplus \begin{cases}(1,0) & \mbox{if $\omega_1(\mlambda i.(1,\hat{a}_i))<|a|$}\\ r & \mbox{otherwise}\end{cases}\\
\mbox{where }r&:=(6+g_0(|a|)+\phi_0(\omega,g,a\ast g_1(|a|)),1+\phi_1(\omega,g,a\ast g_1(|a|)))
\end{aligned}
\end{equation*}
We now solve the complexity component of these equations to find a closed form expression for the complexity of the search functional. For any function $f:\NN\to\NN$ and $n\in\NN$ define $f^n:\NN\to\NN\times\NN$ by
\begin{equation*}
f^n:=\mlambda i.(1,\widehat{\initSeg{f}{n}_i}).
\end{equation*}
We now consider $\phi_0(\omega,g,\initSeg{g_1}{n})$. In the case that $\omega_1(g_1^n)<n$ we have
\begin{equation*}
\phi_0(\omega,g,\initSeg{g_1}{n})=4+\omega_0(g_1^n)+1
\end{equation*}
and otherwise we have
\begin{equation*}
\phi_0(\omega,g,\initSeg{g_1}{n})=4+\omega_0(g_1^n)+6+g_0(n)+\phi_0(\omega,g,\initSeg{g_1}{n+1})
\end{equation*}
Putting these together we obtain
\begin{equation*}
\phi_0(\omega,g,\initSeg{g_1}{n})=5+\omega_0(g_1^n)+\begin{cases}0 & \mbox{if $\omega_1(g_1^n)<n$}\\ 5+g_0(n)+\phi_0(\omega,g,\initSeg{g_1}{n+1}) & \mbox{otherwise}\end{cases}
\end{equation*}
We can then use this to expand $\phi_0(\omega,g,[])$ until we reach a point $N\in\NN$ such that $\omega_1(g_1^N)<N$. More precisely, if $N$ is the first such point, we have
\begin{equation*}
\phi_0(\omega,g,[])=10 N + 5 + \sum_{i=0}^N \omega_0(g^i_1)+\sum_{i=0}^{N-1}g_0(i).
\end{equation*}
This then forms a closed expression for the cost $c$ of evaluating $\spec\; v_1\; v_2\; \varepsilon$ whenever $v_1\btl \omega$ and $v_2\btl g$. Note that this cost expression would be formally defined using the Spector search function itself, just as the cost expression for primitive recursion is also a primitive recursive functional.

\end{exa}

\subsection{Bounded costs}
\label{sec-cost-bound}

In the final part of the paper, we show how we can modify our denotational cost semantics to provide \emph{upper bounds} on the cost of derivations, along the lines of \cite{DanPayRoy(2013.0)}, which is generalised to a richer language with recursion over arbitrary datatypes in \cite{DanLicRam(2015.0)}. The main motivation for looking for upper bounds (rather than a precise measure of complexity) is that it allows us to abstract away certain parts of the program (for instance, treating all number inputs as the same) and thereby obtain simplified expressions for the cost of programs. We now no longer denote a numeral $\num{n}$ by the corresponding natural number $n$, but assign all numerals a uniform size $1$. Lists are then interpreted by a single natural number that represents an upper bound on their length. In this way, we sacrifice precision for a simplified upper bound on the cost of running a program.

To be more precise, while we still have $C=\NN$ we alter the semantic side of our denotational semantics by setting $D_\nat=D_{\nat^\ast}=\NN$, and adapt our logical relation so that
\begin{equation*}
\begin{aligned}
e\Downarrow^c v&:\Leftrightarrow (\exists c')(e\downarrow^{c'} v\wedge c'\leq c)\\
\num{n}\wtl_\nat m&:\Leftrightarrow 1\leq m\\
\num{a}\wtl_{\nat^\ast}m &:\Leftrightarrow |a|\leq m
\end{aligned}
\end{equation*}
and thus $\btl_\rho$ becomes
\begin{equation*}
\begin{aligned}
e\btl_\rho (c,a)&:\Leftrightarrow (\exists v,c')(e\downarrow^{c'}_\rho v\wedge c'\leq c\wedge v\wtl_\rho a)\\
u\wtl_{\rho\to\tau} f&:\Leftrightarrow (\forall v,a)(v\wtl_\rho a\Rightarrow uv\btl_\tau f(a)).
\end{aligned}
\end{equation*}
It is easy to see that $\epsilon:=0$, $\inc(c):=c+1$ and $\comp(c_0,c_1,c_2):=c_0+c_1+c_2$ as defined in Section \ref{sec-cost-exact} is also compatible with our new definition of $\Downarrow$. However, this time we must interpret our constructors differently: We set $\dss{\zero}:=1$ and $\dss{\suc}(n):=1$, so that $\ds{\mt{\num{n}}}=(0,1)$ for all $n\in\NN$, and for lists we define $\dss{\varepsilon}:=0$ and $\dss{::}(m,n):=m+1$, so that $\ds{\mt{\num{a}}}=(0,|a|)$ for all $a\in\NN^\ast$. In this way, we ensure that our constructors are compatible with $\wtl$, which then gives rise to the following metatheorem:
\begin{thm}
\label{thm-bound}
Suppose that each function symbol $f$ of our target language is interpreted by some suitable $\dss{f}$ which satisfies $fv_1\ldots v_k\btl \dss{f}a_1\ldots a_k$ whenever $v_i\wtl_{\rho_i}a_i$. Then for any closed $e:\rho$ we have
\begin{equation*}
e\downarrow^{c}v\mbox{ \; for some \; }c\leq \pi_0\ds{\mt{e}}
\end{equation*}
and some value $v:\rho$.
\end{thm}
We now demonstrate how this metatheorem can be applied to our list based variant of System T. In order to do this we need to generalise the usual maximum operator between two natural numbers to arbitrary types.
\begin{defi}
\label{def-max}
For types $\rho$ of our target language, define $\vee_\rho:\ds{\mt{\rho}}\times\ds{\mt{\rho}}\to \ds{\mt{\rho}}$ inductively by
\begin{equation*}
\begin{aligned}
m\vee_\delta n&:=\max\{m,n\} \mbox{ for $\delta=\nat,\nat^\ast$}\\
f\vee_{\rho\to\tau} g&:=\mlambda a\; . \; (\max\{f_0a,g_0a\},f_1a\vee_\tau g_1a).
\end{aligned}
\end{equation*}
\end{defi}
\begin{lem}
\label{lem-max}
For all types $\rho$, if $u\wtl_\rho a$ then $u\wtl_\rho a\vee_\rho b$ and $u\wtl_\rho b\vee_\rho a$ for any $b\in\ds{\mt{\rho}}$.
\end{lem}

\begin{proof}
Induction on types. We only prove $u\wtl a\vee b$ because the other way round is identical. For base types $\num{n }\wtl m$ implies $1\leq m$, but since then $1\leq \max\{m,k\}$ we have $\num{n}\wtl m\vee_\nat k$. Similarly for lists: $\num{a}\wtl_{\nat^\ast} m$ implies $|a|\leq m\leq\max\{m,k\}$ and thus $\num{a}\wtl_{\nat^\ast} m\vee_{\nat^\ast} k$.

For function types, suppose that $u\wtl_{\rho\to\tau} f$, which means that for any $v\wtl_\rho a$ we have $uv\downarrow^c w$ for some $c,w$ with $c\leq f_0a$ and $w\wtl_\tau f_1a$. But $c\leq \max\{f_0a,g_0a\}$ and by the induction hypothesis $w\wtl_\tau f_1a\vee_\tau g_1a$, and since $v,a$ were arbitrary we have $u\wtl_{\rho\to\tau} f\vee_{\rho\to\tau} g$. 
\end{proof}

\begin{lem}
\label{lem-boundfold}
We have $\fold\; v_1\; v_2\; \num{a}\btl \dss{\fold}(b,h,n)$ for any $v_1\wtl b$, $v_2\wtl h$ and $a\in\NN^\ast$ with $|a|\leq n$, where $\dss{\fold}$ is defined by
\begin{equation*}
\dss{\fold}(a,h,n)=1\dplus\begin{cases}(0,a) & \mbox{if $n=0$}\\ (h1\circ\dss{\fold}_0(a,h,n'),a\vee_\rho h1\circ\dss{\fold}_1(a,h,n')) & \mbox{if $n=n'+1$}\end{cases}
\end{equation*}
\end{lem}

\begin{proof}
Induction on $n$. If $n=0$ then we must have $\num{a}=[]$, and since $v_1\downarrow^0 v_1\wtl_\rho a$ we have $\fold\; v_1\; v_2\; []\downarrow^1 v_1\wtl a$ and thus $\fold\; v_1\; v_2\; []\btl_\rho (1,a)$.

For $n=n'+1$ there are two possibilities. Either $\num{a}=[]$ and $\fold\; v_1\; v_2\; []\downarrow^1_\rho v_1\wtl a$ as before, and since $1\leq 1\dplus h1\circ\dss{\fold}_0(a,h,n')$ and (by Lemma \ref{lem-max}) $v_1\wtl a\vee_\rho h1\circ\dss{\fold}_1(a,h,n')$ we're done. Otherwise $\num{a}=\num{a'}::\num{m}$. By the induction hypothesis we have $\fold\; v_1\; v_2\; \num{a}'\btl_\rho \dss{\fold}(a,h,n')$, and since $\num{m}\wtl_\nat 1$ and thus $v_2\; \num{m}\btl_{\rho\to\rho} h1$ we have (by Lemma \ref{lem-compapp}) $v_2\; \num{m}\; (\fold\; v_1\; v_2\; \num{a'})\btl h1\circ\dss{\fold}(a,h,n')$ which is just
\begin{equation*}
v_2\; \num{m}\; (\fold\; v_1\; v_2\; \num{a}')\downarrow^{c'} w\wedge c'\leq c\wedge w\wtl_\rho b
\end{equation*}
for $(c,b)=h1\circ\dss{\fold}(a,h,n')$, and therefore
\begin{equation*}
\fold\; v_1\; v_2\; \num{a}'::\num{m}\downarrow^{c'+1} w\wedge c'+1\leq c+1\wedge w\wtl_\rho b.
\end{equation*}
Thus $\fold\; v_1\; v_2\; \num{a}'::\num{m}\wtl\dss{\fold}(a,h,n'+1)$ follows from another application of Lemma \ref{lem-max}, and we're done.
\end{proof}
As a result of the above lemma, we obtain soundness of the translation for our list based variant of System T, which is analogous to Corollary 3 of \cite{DanPayRoy(2013.0)}:
\begin{cor}
Let $e:\rho$ be a closed term of our list-based variant of System T. Then $e\downarrow^c v$ for some value $v$ and $c\leq \pi_0\ds{\mt{e}}$.
\end{cor}
Naturally, our approach can also be applied to arbitrary rewrite systems whose function symbols have a suitable interpretation, such as forms of recursion over more general data structures, as in \cite{DanLicRam(2015.0)}. Though the usual formulation of bar recursion does not have a bounded semantics of this kind, we could consider finite variants along the lines of \cite{EOP(2011.0),OliPow(2012.2)}. 

\section{Conclusion}
\label{sec-conc}

We have introduced a general monadic translation acting on higher-order functional languages, which combines both a monadic component and a semantic component, where the latter could play the role of a normal denotational semantics, or alternatively something more interesting, such a variant of the majorizability relation. Applications of our translation included a proof that functionals of type level two expressible by closed terms in our language are continuous, together with various denotational cost semantics for functional languages. The emphasis throughout was less on obtaining new results, and more on demonstrating that ideas from a range of different areas - from proof theory to static program analysis - can be brought together under the same framework. Nevertheless, as a side product we presented for the first time a cost analysis of Spector's variant of bar recursion.

So far, our work only applies to functional languages whose function symbols give rise to terminating computations. An obvious next step would be to incorporate \emph{partiality} into our setting, allowing us to reason about potentially non-terminating computations. In the context of normalization via denotational semantics (discussed here in Section \ref{sec-simple-ds}), this has been explored in more generality by Berger \cite{Berger(2006.0)}, where a term $e$ is shown to be strongly normalizing if $\ds{e}\neq \bot$, where now $\ds{\cdot}$ represents a so-called \emph{strict} denotational semantics. This approach also appeals to the notion of a \emph{stratified} rewrite system, where function symbols $f$ are labelled with natural numbers in order to track the number of times they can be rewritten. A similar extension of Danner et al's complexity framework to arbitrary PCF programs is given in \cite{Kim(2016.0)}. It would be interesting to see if an approach along these lines, using stratified rewrite systems together with a partial denotational semantics, could be used to extend our framework to arbitrary PCF programs, without requiring the user to first prove that function symbols have a suitable interpretation.

There is also, naturally, the prospect of working with other monads and modelling other evaluation strategies, such as parallel computation. However, our simple call-by-value framework based on the writer monad is already rich enough to reason about extensional properties such as majorizability and continuity, and in addition allows us to characterise a variety of cost measures for higher order functional programs.

\section*{Acknowledgment}
  \noindent I am extremely grateful to the anonymous referee for their insightful comments and careful reading of the paper, which led to a much improved final version.


\end{document}